\documentclass[nohypdvips,reqno,final]{siamart0516}
\usepackage{amsmath}
\usepackage{amssymb}
\usepackage{amsfonts}
\usepackage{graphicx}
\usepackage{epstopdf}
\usepackage{ulem}
\usepackage{hyperref}

\theoremstyle{plain}
\newtheorem{thm}{Theorem}[section]
\newtheorem{lem}[thm]{Lemma}

\newtheorem{cor}[thm]{Corollary}

\theoremstyle{definition}

\theoremstyle{remark}

\usepackage{tensor}
\usepackage{MnSymbol}
\usepackage{tikz-cd}

\newcommand{\ie}{i.\,e.}%
\newcommand{\eg}{e.\,g.}%
\newcommand{\st}{s.\,t.}%
\newcommand{\wrt}{w.r.t.}%

\newcommand{\formComma}{\,\text{,}}
\newcommand{\formPeriod}{\,\text{.}}

\newcommand{\R}{\mathbb{R}}%

\newcommand{\landau}{\mathcal{O}}

\newcommand{\surf}{\mathcal{S}}
\newcommand{\surfh}{\surf_{h}}
\newcommand{\normal}{\boldsymbol{\nu}}
\newcommand{\tangent}{\textup{T}}

\newcommand{\gb}{\boldsymbol{g}}%
\newcommand{\Gb}{\boldsymbol{G}}%
\newcommand{\tb}{\boldsymbol{t}}%
\newcommand{\xb}{\boldsymbol{x}}%
\newcommand{\Xb}{\boldsymbol{X}}%
\newcommand{\qb}{\boldsymbol{q}}%
\newcommand{\Qb}{\boldsymbol{Q}}%
\newcommand{\Psib}{\boldsymbol{\Psi}}%
\newcommand{\psib}{\boldsymbol{\psi}}%
\newcommand{\Wb}{\boldsymbol{W}}%
\newcommand{\Eb}{\boldsymbol{E}}%
\newcommand{\Pb}{\boldsymbol{P}}%

\newcommand{\dV}{\text{d}V}%
\newcommand{\dS}{\text{d}\surf}%
\newcommand{\dxi}{\text{d}\xi}%

\newcommand*{\rom}[1]{\textup{\uppercase\expandafter{\romannumeral#1}}}

\newcommand{\shapeOperator}{\boldsymbol{B}}%
\newcommand{\thirdFundamentalForm}{\shapeOperator^{2}}%
\newcommand{\gaussianCurvature}{\mathcal{K}}%
\newcommand{\meanCurvature}{\mathcal{H}}%
\newcommand{\kronecker}{\boldsymbol{\delta}}%

\DeclareRobustCommand{\GGamma}{\text{\raisebox{\depth}{\scalebox{1}[-1]{$\mathbb{L}$}}}}

\newcommand{\Div}{\operatorname{div}}
\newcommand{\Rot}{\operatorname{rot}}

\newcommand{\ch}[2]{\Gamma_{#1}^{#2}}

\newcommand{\trace}{\operatorname{tr}}%

\newcommand{\Proj}[1]{\operatorname{\Pi}_{#1}}
\newcommand{\ProjBound}{\Proj{}}
\newcommand{\ProjSurf}{\Proj{}}

\newcommand{\ProjQ}{\operatorname{P}_{\!\mathcal{Q}}}

\newcommand{\frakk}{\mathfrak{k}}
\newcommand{\frakl}{\mathfrak{l}}

\newcommand{\atsurf}[1]{\left. #1 \right|_{\surf}}

\newcommand{\betashift}[2]{\sigma_{\!#1}\!\left( #2 \right)}

\title{Nematic liquid crystals on curved surfaces --- a thin film limit}

\author{Ingo Nitschke\footnotemark[1]\ \footnotemark[2]
\and Michael Nestler\footnotemark[2]
\and Simon Praetorius\footnotemark[2]
\and Hartmut L\"owen\footnotemark[3]
\and Axel Voigt\footnotemark[2]\ \footnotemark[4]}

\date{\today}

\begin{document}




\maketitle

\renewcommand{\thefootnote}{\fnsymbol{footnote}}
\footnotetext[1]{Corresponding author: ingo.nitschke@tu-dresden.de (Ingo Nitschke)}
\footnotetext[2]{Institut f{\"u}r Wissenschaftliches Rechnen, Technische Universit{\"a}t Dresden, 01062 Dresden, Germany}
\footnotetext[3]{Institut f\"ur Theoretische Physik II - Soft Matter, Heinrich-Heine-Universit\"at D\"usseldorf, 40225 D\"usseldorf, Germany}
\footnotetext[4]{Dresden Center for Computational Materials Science (DCMS), 01062 Dresden, Germany}
\renewcommand{\thefootnote}{\arabic{footnote}}

\begin{abstract}
We consider a thin film limit of a Landau-de Gennes Q-tensor model. In the limiting process we observe a continuous transition where the normal and tangential parts of the Q-tensor decouple and various intrinsic and extrinsic contributions emerge. Main properties of the thin film model, like uniaxiality and parameter phase space, are preserved in the limiting process. For the derived surface Landau-de Gennes model, we consider an $L^2$-gradient flow. The resulting tensor-valued surface partial differential equation is numerically solved to demonstrate realizations of the tight coupling of elastic and bulk free energy with geometric properties.
\end{abstract}

%

\section{Introduction}
We are concerned with nematic liquid crystals whose molecular orientation is subjected to a tangential anchoring on a curved surface. Such surface nematics offer a non trivial interplay between the geometry and the topology of the surface and the tangential anchoring constraint which can lead to the formation of topological defects. An understanding of this interplay and the resulting type and position of the defects is highly desirable.

As an application, nematic shells have been proposed as switchable capsules optimal for a steered drug delivery \cite{Jenekhe1998}. The defect structure thereby essentially determines where the shells can be opened in a minimal destructive way. Moreover, nematic shells are possible candidates to form supramolecular building blocks for tetrahedral crystals with important implications for photonics \cite{Zerrouki2006}.

Besides such equilibrium structures, defects also play a fundamental role in active systems. In \cite{Keber2014} the spatiotemporal patterns that emerge when an active nematic film of microtubules and molecular motors is encapsulated within a lipid vesicle is analyzed. The combination of activity, topological constraints, and geometric properties produces a myriad of dynamical states. Understanding these relations offers a way to design biomimetic materials, with topological constraints used to control the non-equilibrium dynamics of active matter.

Defects in nematic shells are intensively studied on a sphere \cite{Dzubiella2000,Bates2010,Shin2008,LopezLeon2011,Leon2011,Dhakal2012,Koning2013} and under more complicated constraints, see, \eg, \cite{Prinsen2003,Selinger2011,Nguyen2013,Martinez2014,Segatti2014,Alaimo2017}. However, most of these studies use particle methods. Despite the interest in such methods a continuum description would be more essential for predicting and understanding the macroscopic relation between position and type of the defects and geometric properties of the surface. For bulk nematic liquid crystals the Landau-de Gennes Q-tensor theory \cite{Virga1994,Stewart2004} is a well established field theoretical description. For a mathematical review we refer to \cite{Ball2017}. However, its surface formulation is still under debate. Surface models have been postulated by analogue derivations on the surface \cite{Kralj2011}, by considering the limit of vanishing thickness for bulk Q-tensors models \cite{Napoli2012b,Golovaty2017} or via a discrete-to-continuum limit \cite{Canevari2017}. The derived models differ in details and strongly depend on the made assumptions in the derivation.

Our approach aims to derive a surface Q-tensor model by dimensional reduction via a thin film limit of a general bulk Landau-de Gennes model. In contrast to previous work we only make assumptions on the boundary of the thin film where we admit only states conforming to critical points of the free energy.  In the limiting process we observe a continuous transformation where the normal and tangential parts of the Q-tensor decouple and various intrinsic and extrinsic contributions emerge. The obtained surface Landau-de Gennes energy is compared with previous models \cite{Kralj2011,Napoli2012b,Golovaty2017,Canevari2017} and an $L^2$-gradient flow is considered. The resulting tensor-valued surface partial differential equation is solved numerically on an ellipsoid.

The paper is structured as follows. In Section 2 we present the main results, including the surface Landau-de Gennes energy, a formulation for the evolution
problem, and numerical results to illustrate the mentioned interplay between the geometry, the topology of the surface, and the positions and type of the defects. Section 3 establishes the notation
essential for the derivation of the thin film limit, which is derived in Section 4 for the energy and the $L^2$-gradient flow. A discussion of mathematical and physical implications of the derived model and a comparison with previously postulated thin film models is provided in Section 5. Conclusions are drawn in Section 6 and details of the analysis are given in the Appendix.

\section{Main results}
We consider Q-tensor fields on Riemannian manifolds \( \mathcal{M} \) defined by \( \mathcal{Q}(\mathcal{M}):= \{ \tb\in\tangent^{(2)}(\mathcal{M})\,:\, \trace{\tb} = 0,\,\tb = \tb^{T}\}. \)
We assume \( \mathcal{M} \) as well as \( n \)-tensor bundles \( \tangent^{(n)}(\mathcal{M}) \) to be sufficiently smooth and consider two types of manifolds \( \mathcal{M} \), a surface $\surf$ without boundaries and a thin film $\surf_h:=\surf\times[-h/2,h/2]$ of thickness $h$. We have \( \mathcal{Q}(\surf) \subset \atsurf{\mathcal{Q}(\surfh)} \)
and we can tie a surface Q-tensor \( \qb\in\mathcal{Q}(\surf) \) with a restricted bulk Q-tensor \( \Qb\in\atsurf{\mathcal{Q}(\surfh)} \) by the orthogonal projections \( \ProjSurf= \operatorname{Id} - \normal\otimes\normal \), with identity $\operatorname{Id}$ and surface normal $\normal$ and \( \ProjQ \) a Q-tensor projection defined in \eqref{eq:qtensorprojection}, \ie,
\begin{equation}\label{eq:surface-Q-tensor}
  \qb = \ProjQ\left( \ProjSurf\Qb|_\surf\ProjSurf \right) = \ProjSurf\Qb|_\surf\ProjSurf + \frac{1}{2}\left( \normal\Qb|_\surf \normal\right)\ProjSurf\formPeriod
\end{equation}
For Q-tensors $\Qb\in\mathcal{Q}(\surf_h)$ we consider the elastic and bulk free energy $\mathcal{F}^{\surf_h} = \mathcal{F}^{\surf_h}_\text{el} + \mathcal{F}^{\surf_h}_\text{bulk}$ with
\begin{equation}\begin{split}\label{eq:ThinFilmEnergy}
  \mathcal{F}^{\surf_h}_\text{el}[\Qb] &:= \frac{1}{2}\int_{\surf_h} L_{1} \|\nabla\Qb\|^2
                                                                   + L_{2} \left\| \Div\Qb \right\|^{2}
                                                                   + L_{3} \left\langle \nabla\Qb , \left( \nabla\Qb \right)^{T_{(2\,3)}} \right\rangle
                                                                   + L_{6} \left\langle \left( \nabla\Qb \right)\Qb , \nabla\Qb \right\rangle\,\dV\,, \\
  \mathcal{F}^{\surf_h}_\text{bulk}[\Qb] &:= \int_{\surf_h} a\trace \Qb^2 + \frac{2}{3}b\trace \Qb^3 + c\trace \Qb^4\,\dV\,,
\end{split}\end{equation}
see, \eg, \cite{Mottram2014}, with elastic parameters \( L_{i} \)
and thermotropic parameters $a,b$, and $c$. For simplicity we restrict our analysis to achiral liquid crystals, i.e. $L_4 = 0$, see the general form in \cite{Mottram2014}.

Moreover, let \( \ProjBound\Qb\normal = \normal\Qb\ProjBound = 0 \) and \( \normal\Qb\normal = \beta \) be essential anchoring conditions at \( \partial\surf_{h} \),
where \( \beta \) is considered to be constant. Consequently, we obtain the natural anchoring conditions
\begin{align}\label{eq:NatBC}
  \ProjBound\left( \left( L_{1} + \beta L_{6} \right) \left(\nabla\Qb\right)\normal + L_{3}  \left(\nabla\Qb\right)^{T_{(2\,3)}}\normal \right)\ProjBound &=0 \quad\mbox{ at } \partial\surf_{h}
\end{align}
which ensure vanishing boundary integrals in the first variation \( \delta\mathcal{F}^{\surf_h} \). For $\qb$ as in \eqref{eq:surface-Q-tensor} we obtain in the thin film limit $\frac{1}{h}\mathcal{F}^{\surf_h}[\Qb] = \mathcal{F}^{\surf}[\qb] + \mathcal{O}(h^2)$ the corresponding surface free energy $\mathcal{F}^{\surf} = \mathcal{F}^{\surf}_\text{el} + \mathcal{F}^{\surf}_\text{bulk}$ with
\begin{equation}\begin{split}\label{eq:SurfaceEnergy}
  \mathcal{F}^{\surf}_\text{el}[\qb] &:= \frac{1}{2}\int_\surf  L'_{1} \|\nabla\qb\|^2
                                                              + L_{6} \left\langle \left( \nabla\qb \right)\qb , \nabla\qb \right\rangle \\
                                     &\quad\quad + M_{1}\trace\qb^{2}
                                                 + M_{2}\left\langle \shapeOperator,\qb \right\rangle^{2}
                                                 + M_{3}\trace\qb^{2} \left\langle \shapeOperator,\qb \right\rangle
                                                 + M_{4}\left\langle\shapeOperator,\qb \right\rangle
                                                + C_{0} \,\dS, \\
  \mathcal{F}^{\surf}_\text{bulk}[\qb] &:= \int_\surf a'\trace\qb^2 + c\trace\qb^4 + C_{1}\,\dS,
\end{split}\end{equation}
and shape operator $\shapeOperator=-(\ProjSurf\nabla)\normal$.
In contrast to \eqref{eq:ThinFilmEnergy}, all operators are defined by the Levi-Civita connection and inner products are considered at the surface.
All parameter functions $L'_{1},M_{1},M_{2},M_{3},M_{4},C_{0},C_{1},$ and $a'$ can be related to the thin film parameters $L_{i}$, the surface quantities
$\meanCurvature$ (mean curvature) and $\gaussianCurvature$ (Gaussian curvature), and $\beta$, see \eqref{eq:parameterfunctions}. The \( L^{2} \)-gradient flow $\partial_t \qb = - \nabla_{L^2} \mathcal{F}^{\surf} $ reads
\begin{align}\label{eq:evoliEquation}
  \partial_{t}\qb &=
       L'_{1} \Delta^{dG}\qb
       + L_{6}\left(\left( \nabla\nabla\qb \right):\qb + \left( \nabla\qb \right)\cdot\Div\qb
                    - \frac{1}{2}\left( \nabla\qb \right)^{T_{(1\,3)}}:\nabla\qb + \frac{1}{4}\left\| \nabla\qb \right\|^{2}\gb \right)\notag\\
     &\quad  -\left( M_{1} + M_{3}\left\langle \shapeOperator,\qb \right\rangle  + 2a' +2c \trace\qb^{2} \right)\qb
      -\left( M_{2}\left\langle \shapeOperator,\qb \right\rangle + \frac{M_{3}}{2}\trace\qb^{2} + \frac{M_{4}}{2} \right)\left(\shapeOperator - \frac{1}{2}\meanCurvature\gb\right)
\end{align}
on $\surf \times [0,T]$ with the div-Grad (Bochner) Laplacian $ \Delta^{dG}$. The same evolution equation also follows as the thin film limit of the corresponding $L^2$-gradient flow
$\partial_t \Qb = - \nabla_{L^2} \mathcal{F}^{\surf_h}$ for \eqref{eq:ThinFilmEnergy}.

\begin{figure}[ht]
\begin{center}
\includegraphics[width=.23\linewidth]{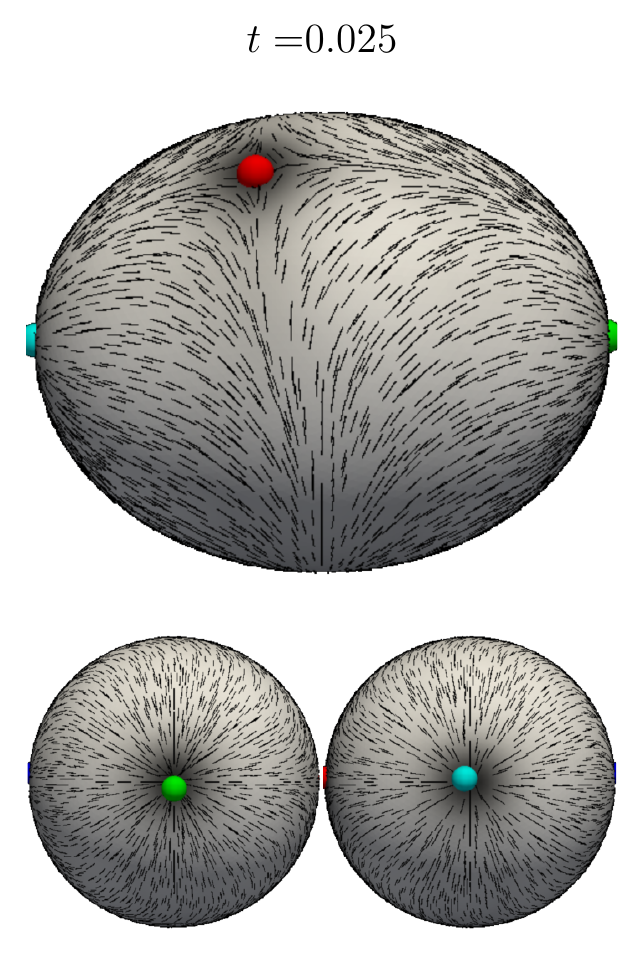}
\includegraphics[width=.23\linewidth]{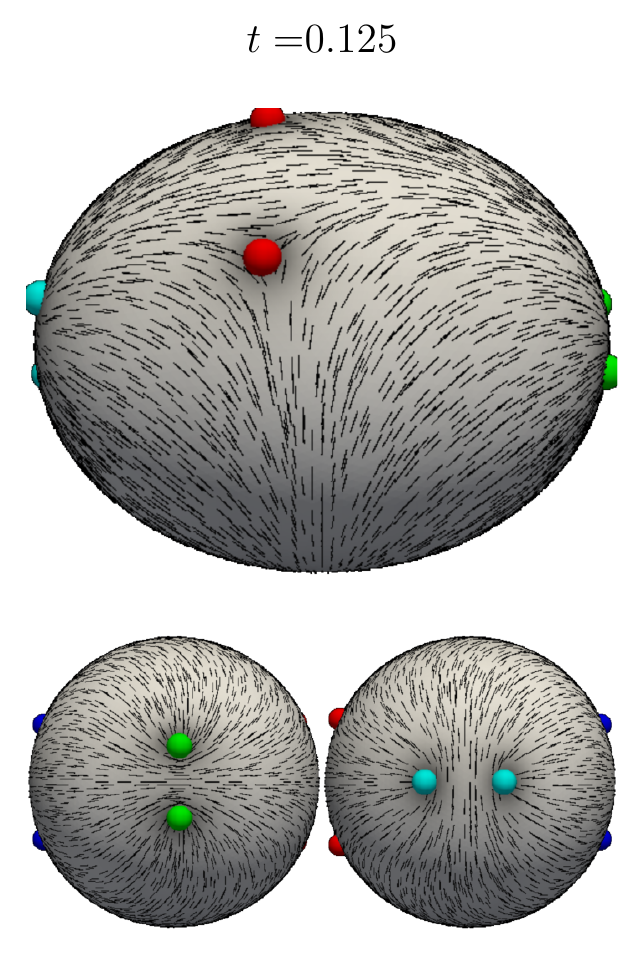}
\includegraphics[width=.23\linewidth]{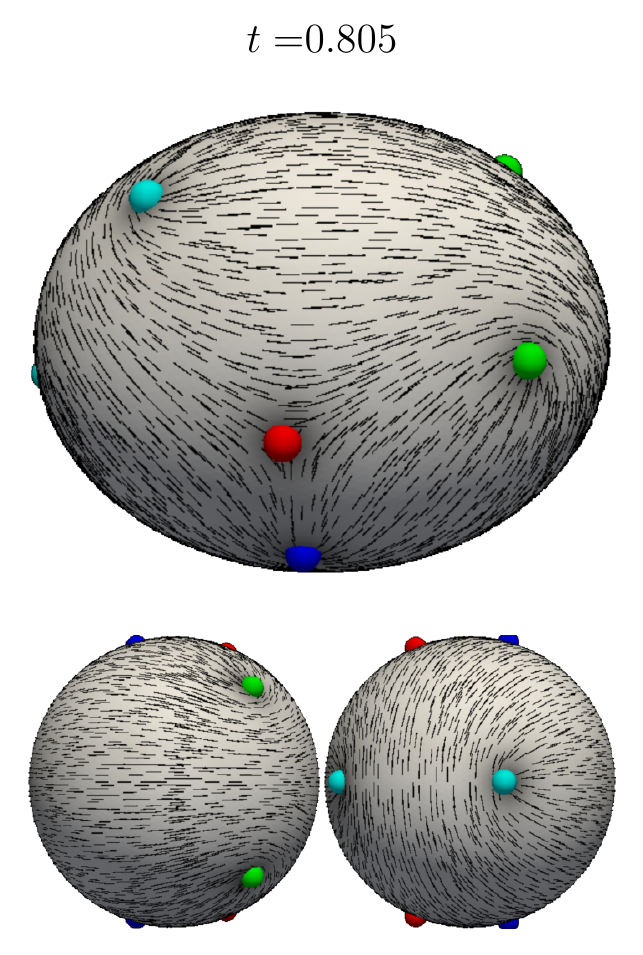}
\includegraphics[width=.23\linewidth]{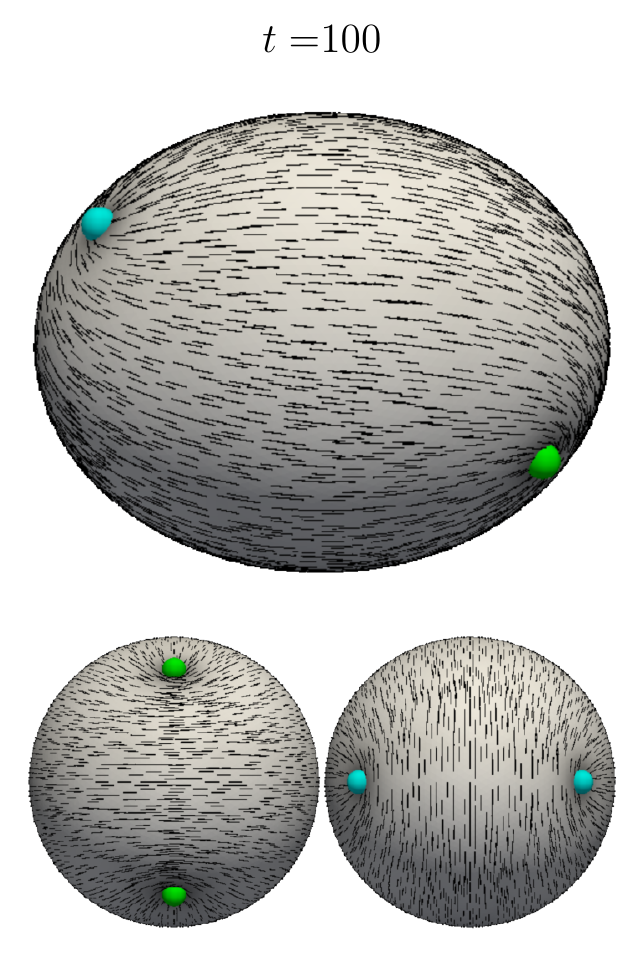}\\
\includegraphics[width=.96\linewidth]{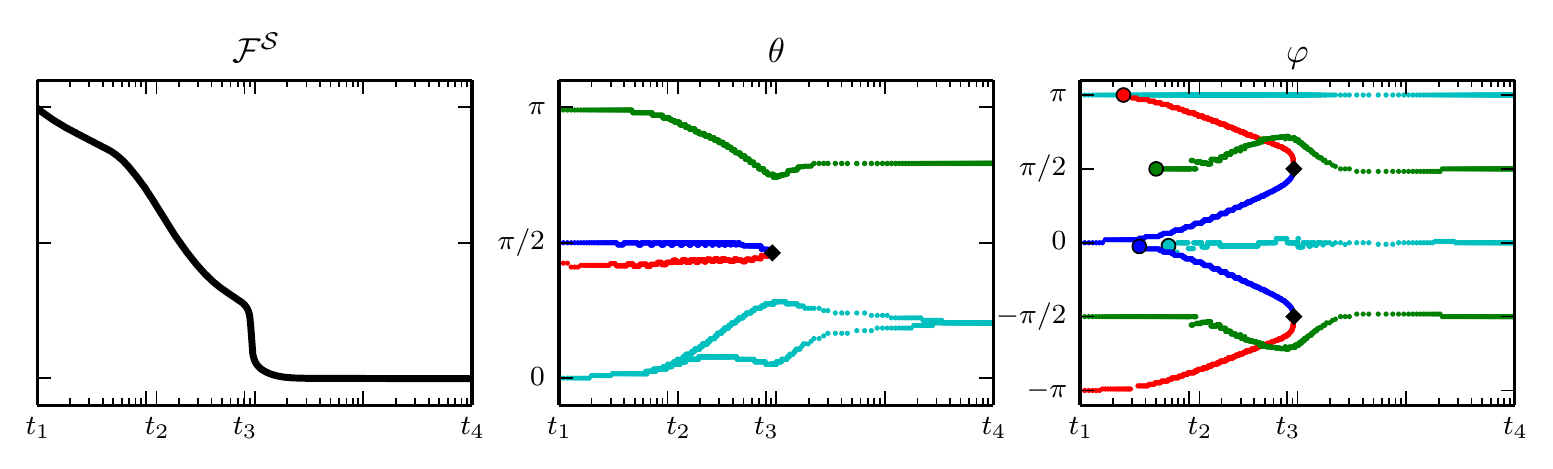}
\end{center}
\caption{\textbf{Numerical simulation on an ellipsoid:} (colors online) (top) snapshots [side, top, bottom view] of defects and principal director. From left to right: initial state of four defects with topological charge $+1$ (nodes: blue, green, cyan) and $-1$ (saddles: red); break down into pairs of $+ 1/2$ (wedges) and $- 1/2$ (trisectors) defects, respectively; attraction and repulsion of defects until two pairs of oppositely charged defects annihilate; minimum energy state with four $+1/2$ defects in deformed tetrahedral configuration as described in \cite{Kralj2011}. (bottom) From left to right: surface free-energy $\mathcal{F}^{\surf}$ over time; defect positions in spherical coordinates with polar angle $\theta$ and azimuthal angle $\varphi$ over time. The colors correspond to the marked defects in the top row. Colored dots mark emerging defects, black diamonds indicate defect annihilation. The half-axis of the ellipsoid are $[1,1,1.25]$ and parameters are $L_1 = L_2 = - L_3 = 1$, $L_6 = 0$, $M_2 = M_3 = 0$ and $a = - 2/3$, $b = -1/2$, $c = 1$. We further consider $\mathcal{F}^{\surf} = \mathcal{F}^{\surf}_\text{el} + \omega \mathcal{F}^{\surf}_\text{bulk}$ with $\omega = 100$.}
\label{fig:defectDynamics}
\end{figure}

To numerically solve the tensor-valued surface partial differential equation \eqref{eq:evoliEquation} we use a similar approach as considered in
\cite{Nestler2017,Reuther2017}. We reformulate the equation in $\R^3$ euclidean coordinates and penalize all normal contributions $\qb \cdot \normal$, to enforce tangentiality of the tensor. This leads to a coupled nonlinear system
of scalar-valued surface partial differential equations for the components of $\qb$, which can be solved using surface finite elements \cite{Dziuk2013}. \autoref{fig:defectDynamics} shows the evolution on a spheroidal ellipsoid. The initial configuration is set as in \cite{Nestler2017}(Fig. 1), with three nodes and a saddle defect, which are placed along an equatorial plane. In accordance with the Poincar\'e-Hopf theorem the topological charges of these defects add up to the Euler characteristic of the surface, $1+1+1-1 = 2$. After some rearrangement all four defects split into pairs of $+ \frac{1}{2}$ and $- \frac{1}{2}$ defects, which move away from each other perpendicular to the initial equatorial plane. Equally charged defects repel each other and oppositely charged defects attract each other. This leads to an annihilation of two pairs of $+\frac{1}{2}$ and $-\frac{1}{2}$ defects. According to the geometric properties of the ellipsoid the remaining four $+\frac{1}{2}$ defects arrange pairwise in the vicinity of the high curvature regions, with each pair perpendicular to each other. This deformed tetrahedral configuration is known to be the minimal energy state, see \cite{Lubensky1992,Nelson2002} for a sphere and \cite{Kralj2011} for ellipsoids. We further observe the principle director to be aligned with the minimal curvature lines in the final configuration. This alignment is a consequence of the extrinsic contributions in \eqref{eq:SurfaceEnergy}, where our model differs from previous studies. Another remarkable feature of the derived surface Landau-de Gennes model is the possibility of coexisting isotropic and nematic phases. Such coexistence is know in three-dimensional models and results from the presence of the $\trace\Qb^3$ term in \eqref{eq:ThinFilmEnergy}. Such a term is absent in two-dimensional models in flat space. This difference in the three- and two-dimensional model typically changes the phase transition type. In our model the dependency of $M_1$ on curvature, see \eqref{eq:parameterfunctions}, allows to locally modify the double-well potential in \eqref{eq:SurfaceEnergy} and thus allows for coexisting states due to changing geometric properties of the surface.

\section{Notational convention and thin film calculus}
For notational compactness of tensor algebra we use the Ricci calculus, where lowercase indices $i,j,k,\ldots$
denote components in a surface coordinate system and uppercase indices $I,J,K,\ldots$ denote
components in the extended three dimensional thin film coordinate system.
Brackets $[]$ and $\{\}$ are used to switch between components and object representation, \ie, for a
2-tensor $\tb$ we write $[\tb]_{ij}=t_{ij}$ for the components and $\{t_{ij}\}=\tb$
for the object.
Most of the tensor formulations in this paper are invariant \wrt\ coordinate transformations,
thus a co- and contravariant distinction in the object representation is not necessary.
However, if such a distinction is needed, we use
the notation of musical isomorphisms $\sharp$ and $\flat$ for raising and lowering indices,
respectively.
These are extended to tensors in a natural way, \eg, for a 2-tensor $\tb=\{\tensor{t}{^{i}_{j}}\}\in\tensor{\tangent}{^{1}_{1}}\surf$ we write
${}^{\flat}\tb^\sharp = \{\tensor{t}{_{i}^{j}}\} = \gb\{\tensor{t}{^{i}_{j}}\}\gb^{-1}\in\tensor{\tangent}{_{1}^{1}}\surf$ with metric tensor \( \gb \) in \( \surf \). Finally, a tensor product
denotes a contraction $[\boldsymbol{s}\tb]_{ij}:=\tensor{s}{_{i}^{k}}\tensor{t}{_{k}_{j}}$ and the Frobenius
norm of a rank-$n$ tensor $\tb$ will be denoted by $\|\tb\|_{\gb}$, \ie, $\|\tb\|_{\gb}^2=\langle\tb,\,\tb\rangle_{\gb}$ with
$\langle\boldsymbol{s},\,\tb\rangle_{\gb} := \tensor{s}{_{i_1\cdots i_n}}\tensor{t}{^{i_1\cdots i_n}}$, that has to be understood
\wrt\ the corresponding metric $\gb$ for raising and lowering the indices\footnote{The suffix $\gb$ will be omitted, if it is
clear which metric the scalar product refers to.}.

The first, second, and third fundamental form are denoted by $g_{ij} = \langle\partial_i\xb,\,\partial_j\xb\rangle$ (metric tensor),
$[\shapeOperator]_{ij} = -\langle\partial_i\xb,\,\partial_j\normal\rangle$ (covariant shape operator), and
$[\thirdFundamentalForm]_{ij} = \langle\partial_i\normal,\,\partial_j\normal\rangle$, respectively. With
this, curvature quantities can be derived: $\gaussianCurvature = \operatorname{det}\shapeOperator^\sharp$ (Gaussian curvature) and
$\meanCurvature = \operatorname{tr}\shapeOperator=\tensor{\shapeOperator}{^{i}_{i}}$ (mean curvature). The Kronecker
delta will be denoted by $\kronecker = \{\delta^i_j\}$ and the Christoffel symbols (of second kind)
will be denoted by $\ch{ij}{k}=\frac{1}{2}g^{kl}(\partial_i g_{jl} + \partial_j g_{il} - \partial_l g_{ij})$ at the surface and
$\GGamma_{IJ}^{K}=\frac{1}{2}G^{KL}(\partial_I G_{JL} + \partial_J G_{IL} - \partial_L g_{IJ})$ in the thin film,
where \( \Gb \) is the metric tensor of the thin film \( \surfh \),
\eg, \( G_{IJ} = \delta_{IJ} \) and \( \GGamma_{IJ}^{K} = 0 \) in the euclidean case.

The surface \( \surf \) and the thin film \( \surfh \) as Riemannian manifolds are equipped with different metric compatible Levi-Civita connections \( \nabla \).
We use ``$;$'' in the thin film and ``$|$'' at the surface to point out the difference for covariant derivatives in index notation, \eg,
\begin{align}
  \left[ \nabla\Qb \right]_{IJK} &= Q_{IJ;K} = \partial_{K}Q_{IJ} - \GGamma_{KI}^{L}Q_{LJ} - \GGamma_{KJ}^{L}Q_{IK} && \text{in }\surfh\formComma\label{eq:CoDerivativeThin}\\
  \left[ \nabla\qb \right]_{ijk} &= q_{ij|k} = \partial_{k}q_{ij} - \Gamma_{ki}^{l}q_{lj} - \Gamma_{kj}^{l}q_{ik} && \text{in }\surf\formPeriod\label{eq:CoDerivativeSurf}
\end{align}

We define the coordinate in normal direction \( \normal \) of the surface \( \surf \) by \( \xi\in\left[ -\frac{h}{2}, \frac{h}{2} \right] \).
The local surface coordinates are \( (u,v) \) defined on every chart in the atlas of \( \surf \), \st\ the immersion
\( \xb: (u,v) \mapsto \R^{3} \) parametrize the surface.
Adding these up, we obtain a parametrization \( \Xb: (u,v,\xi) \mapsto \R^{3} \) of the thin film \( \surfh \), defined by
\( \Xb(u,v,\xi) := \xb(u,v) + \xi\normal(u,v) \formPeriod \)
This means, the lowercase indices $i,j,k,\ldots$ are in \( \{u,v\} \) and the uppercase indices $I,J,K,\ldots$ are in \( \{u,v,\xi\} \).
The canonical choice of basis vectors in the tangential bundles are \( \partial_{i}\xb\in\tangent\surf \) and \( \partial_{I}\Xb\in\tangent\surfh \).
Therefore, the metric tensors are defined by \( g_{ij} = \partial_{i}\xb \cdot \partial_{j}\xb \) and \( G_{ij} = \partial_{I}\Xb \cdot \partial_{J}\Xb \).
Consequently, it holds \( G_{i\xi} = G_{\xi i} = 0  \), \( G_{\xi\xi} = 1 \), and by \eqref{eq:InverseMetric}, we get for the inverse metric tensor
\( G^{i\xi} = G^{\xi i} = 0 \), \( G^{\xi\xi}=1 \). The pure tangential components of the thin film metric and its inverse can be expressed as a
second order surface tensor polynomial in \( \xi\shapeOperator \) and a second order expansion
\begin{align}\label{eq:MetricExpand}
  G_{ij} = g_{ij} - 2\xi B_{ij} + \xi^{2}\left[ \thirdFundamentalForm \right]_{ij} \quad \text{and} \quad G^{ij} &= g^{ij} + 2\xi B^{ij} + \landau(\xi^{2}) \formComma
\end{align}
respectively. Consequently, there is no need for rescaling while lowering or rising the normal coordinate index \( \xi \), \ie\,
for an arbitrary thin film tensor \( \Wb \) it holds
\begin{align}\label{eq:RiseLowNormalIndex}
  \tensor*{W}{*^{\ldots}_{\ldots}^{\xi}_{ }^{\ldots}_{\ldots}}
        &= G^{\xi I} \tensor*{W}{*^{\ldots}_{\ldots}^{}_{I}^{\ldots}_{\ldots}}
         = \tensor*{W}{*^{\ldots}_{\ldots}^{}_{\xi}^{\ldots}_{\ldots}}\formPeriod
\end{align}
Moreover, a contraction of two arbitrary thin film tensor \( \Wb \) and \( \widetilde{\Wb} \) restricted to the surface results in a contraction of the tangential part \wrt\ the surface metric
and a product of the normal part, \ie,
\begin{align}\label{eq:ContractionThinToSurf}
  \begin{aligned}
  \atsurf{\tensor*{W}{*^{\ldots}_{\ldots}^{}_{I}^{\ldots}_{\ldots}}\tensor*{\widetilde{W}}{*^{\ldots}_{\ldots}^{I}_{}^{\ldots}_{\ldots}}}
      &= \atsurf{G^{IJ}\tensor*{W}{*^{\ldots}_{\ldots}^{}_{I}^{\ldots}_{\ldots}}\tensor*{\widetilde{W}}{*^{\ldots}_{\ldots}^{}_{J}^{\ldots}_{\ldots}}}
      = g^{ij}\atsurf{\tensor*{W}{*^{\ldots}_{\ldots}^{}_{i}^{\ldots}_{\ldots}}\tensor*{\widetilde{W}}{*^{\ldots}_{\ldots}^{}_{j}^{\ldots}_{\ldots}}}
       + \atsurf{\tensor*{W}{*^{\ldots}_{\ldots}^{}_{\xi}^{\ldots}_{\ldots}}\tensor*{\widetilde{W}}{*^{\ldots}_{\ldots}^{}_{\xi}^{\ldots}_{\ldots}}} \\
      &= \atsurf{\tensor*{W}{*^{\ldots}_{\ldots}^{}_{i}^{\ldots}_{\ldots}}\tensor*{\widetilde{W}}{*^{\ldots}_{\ldots}^{i}_{}^{\ldots}_{\ldots}}}
       + \atsurf{\tensor*{W}{*^{\ldots}_{\ldots}^{}_{\xi}^{\ldots}_{\ldots}}\tensor*{\widetilde{W}}{*^{\ldots}_{\ldots}^{}_{\xi}^{\ldots}_{\ldots}}}\formPeriod
  \end{aligned}
\end{align}

To deal with covariant derivatives, we have to take the Christoffel symbols into account.
It is sufficient to expand \( \GGamma_{IJ}^{K} \) first order in normal direction, since we only use first order derivatives and no partial derivatives of the symbols are
necessary.
Hence, \eqref{eq:MetricExpand} result in
\begin{align}\label{eq:ChristoffelExpansion}
    \GGamma_{ij}^{k} = \Gamma_{ij}^{k} + \landau(\xi)\formComma \quad
    \GGamma_{ij}^{\xi} = B_{ij} + \landau(\xi) \formComma \quad
    \GGamma_{\xi\xi}^{K} = \GGamma_{I\xi}^{\xi} = \GGamma_{\xi I}^{\xi} = 0\formComma \quad \text{and} \quad
     \GGamma_{i\xi}^{k} = \GGamma_{\xi i}^{k} &= -B_{ij} + \landau(\xi) \formPeriod
\end{align}
The volume element \( \dV \) can be split up into a surface and a normal part by \eqref{eq:DetMetricExpand}, \ie,
\begin{align}\label{eq:VolumeElementExpansion}
  \dV &= \sqrt{\det\Gb}\,\text{d}u\text{d}v\dxi
        = (1 - \xi\meanCurvature + \xi^{2}\gaussianCurvature)\sqrt{\det\gb}\,\text{d}u\text{d}v\dxi
       = (1 - \xi\meanCurvature + \xi^{2}\gaussianCurvature) \dS\dxi \formPeriod
\end{align}

\section{Thin film limit}
Thin film limits require a reduction of degrees of freedom.  We deal with this issue by setting Dirichlet boundary conditions for the normal parts of \( \Qb \) and postulate
a priori a minimum of the free energy on the inner and outer boundary of the thin film. This is achieved by considering natural boundary condition of the weak Euler-Lagrange equation. In this setting we restrict the density of \(  \mathcal{F}^{\surf_h} \) to the surface and integrate in normal direction to obtain the surface energy \( \mathcal{F}^{\surf} \). In the same way, we also show the consistency of the thin film and surface \( L^{2} \)-gradient flows. The next subsection considers the reformulation of the surface Landau-de Gennes energy to obtain the formulation in \eqref{eq:SurfaceEnergy}, which allows a distinction of extrinsic and intrinsic contributions. Finally, we present a strong formulation of the derived equation of motion.

\subsubsection{Derivation of thin film limits}\label{sec:derivationOfThinFilmLimit}
The free energy \eqref{eq:ThinFilmEnergy} in the thin film \( \surfh \) in index notation reads
\begin{align}\label{eq:ThinFilmEnergyIndices}
    \mathcal{F}^{\surf_h}_\text{el}[\Qb] &=
        \frac{1}{2}\int_{\surfh}
            L_{1}Q_{IJ;K}Q^{IJ;K}
          + L_{2}\tensor{Q}{_{I}^{J}_{;J}}\tensor{Q}{^{IK}_{;K}}
          + L_{3}Q_{IJ;K}Q^{IK;J}
          + L_{6}Q^{KL}Q_{IJ;K}\tensor{Q}{^{IJ}_{;L}}
        \dV\notag\\
      \mathcal{F}^{\surf_h}_\text{bulk}[\Qb] &=
        \int_{\surf_h}
          a Q_{IJ}Q^{JI} + \frac{2}{3}b Q_{IJ}Q^{JK}\tensor{Q}{_{K}^{I}} + c Q_{IJ}Q^{JK}Q_{KL}Q^{LI}
        \dV \formPeriod
\end{align}
With respect to arbitrary thin film Q-tensors \( \Psib \in\mathcal{Q}(\surf_h) \),
the corresponding first variations
\begin{align} %
  \delta\mathcal{F}^{\surf_h}_\text{el}\left( \Qb, \Psib \right)
      &=\int_{\surf_h} \Psi_{IJ;K}\left( L_{1}Q^{IJ;K} + L_{3}Q^{IK;J} + L_{6}Q^{KL}\tensor{Q}{^{IJ}_{;L}} \right) \nonumber\\
      &\quad\quad + L_{2}\tensor{\Psi}{_{I}^{J}_{;J}}\tensor{Q}{^{IK}_{;K}}
                     + \frac{L_{6}}{2}\Psi_{IJ}\tensor{Q}{_{KL}^{;I}}Q^{KL;J}\dV\formComma \label{eq:ElFirstVar}\\
  \delta\mathcal{F}^{\surf_h}_\text{bulk}\left( \Qb, \Psib \right)
      &= 2\int_{\surf_h} \left( \left( a + cQ_{KL}Q^{KL} \right)Q^{IJ} + b Q^{IK}\tensor{Q}{_{K}^{J}} \right) \Psi_{IJ} \dV\formPeriod
\end{align}
which are used to find local minimizers of the functional \(\mathcal{F}^{\surf_h} = \mathcal{F}^{\surf_h}_\text{el} +  \mathcal{F}^{\surf_h}_\text{bulk} \),
using the \( L^{2} \)-gradient flow
\begin{align}\label{eq:ELEqThin}
  \int_{\surfh}\left\langle \partial_{t}\Qb, \Psib \right\rangle \dV
   &= -\delta\mathcal{F}^{\surf_h}\left( \Qb, \Psib \right)
        = -\int_{\surfh}\left\langle \nabla_{L^{2}}\mathcal{F}^{\surf_h}, \Psib \right\rangle\dV
\end{align}
for all \( \Psib \in\mathcal{Q}(\surf_h) \).
However, integration by parts of \eqref{eq:ElFirstVar} gives
\begin{align}\label{eq:EulerWithBound}
  \delta\mathcal{F}^{\surf_h}\left( \Qb, \Psib \right)
      &= \int_{\surfh}\left\langle \nabla_{L{2}}\mathcal{F}^{\surf_h}  , \Psib \right\rangle\dV\\
      &\quad  + \int_{\partial\surfh}
          L_{2} \tensor{Q}{_{I}^{J}_{;J}} \tensor{\Psi}{^{I}_{\xi}}
         +\left( L_{1} Q_{IJ;\xi} + L_{3}Q_{I\xi;J} + L_{6}Q^{\xi K}Q_{IJ;K} \right)\Psi^{IJ}
        \text{d}A \notag\formComma
\end{align}
where \( \text{d}A \) is the volume form of the boundary surfaces.
For the choice of essential boundary conditions, we require that \( \Qb \) has to have two directors in the boundary tangential bundle and the remaining director has to be the boundary
normal, \ie, for \( \Pb\in\tangent\partial\surfh \) a pure covariant representation of \( \Qb \) at the boundary is
\begin{align}
  \Qb = S_{1} \Pb^{\flat}\otimes\Pb^{\flat} + S_{2} \normal^{\flat}\otimes\normal^{\flat} - \frac{1}{3}\left( S_{1}+S_{2} \right)\Gb
\end{align}
with scalar order parameter \( S_{1} \) and \( S_{2} \).
Hence, it holds \( Q_{i\xi} = Q_{\xi i} = 0 \) and \( Q_{\xi\xi}= \frac{1}{3} (2S_{2} - S_{1}) \).
For simplicity, we set the pure normal part of \( \Qb \) constant, \ie, \( Q_{\xi\xi}= \beta\in\R \) at \( \partial\surfh \).
Therefore, \( \Psib \) has to be in \(\mathcal{Q}_{0}(\surf_h) := \left\{ \Psib\in \mathcal{Q}(\surf_h):\, \Psi_{I\xi} = \Psi_{\xi I} = 0 \text{ at }\partial\surfh \right\} \),
and we consider the natural boundary conditions
\( 0 = \left( L_{1} + L_{6}\beta \right)Q_{ij;\xi} + L_{3}Q_{i\xi;j}\) at \(\partial\surfh \formComma\)
so that the boundary integral in \eqref{eq:EulerWithBound} vanishs.
Here, our analysis differs from previous results, which deal with a global determination of the normal derivatives in the whole bulk of \( \surfh \) by
parallel transport \( \nabla_{\xi}\Qb = 0 \), or by \( \partial_{\xi}\Qb = 0 \), see \cite{Napoli2012b,Golovaty2017}.

With \autoref{lem:BCAtSurfaceExpansion} we can relate the anchoring conditions to surface identities
\begin{align}\label{eq:AnchoringExpansion}
    \atsurf{Q_{\xi\xi}} &= \beta + \landau(h^{2})
        & \atsurf{\partial_{\xi}Q_{\xi\xi}} &= \landau(h^{2}) & \!\!\!\!\!\!\!\!\!\!\!\!\!\!\left( L_{1} + L_{6}\beta \right)\atsurf{Q_{ij;\xi}} + L_{3}\atsurf{Q_{i\xi;j}} &= \landau(h^{2})\notag \\
    \atsurf{Q_{i\xi}} &= \atsurf{Q_{\xi i}} = \landau(h^{2})
        &\atsurf{\partial_{\xi}Q_{i \xi}} &= \atsurf{\partial_{\xi}Q_{\xi i}} = \landau(h^{2})\formPeriod
\end{align}
Evaluating \( \Psib\in \mathcal{Q}_{0}(\surf_h)\) at the surface results in
\( \atsurf{\Psi_{I\xi}} = \atsurf{\Psi_{\xi I}} =  \atsurf{\partial_{\xi}\Psi_{I \xi}} = \atsurf{\partial_{\xi}\Psi_{\xi I}} =  \landau(h^{2})\formPeriod
\)
The restricted Q-tensor \( \left\{ \atsurf{Q_{ij}} \right\}\in\tangent^{(2)}\surf \)
is not a Q-tensor, because \( \trace_{\gb}\left\{ \atsurf{Q_{ij}} \right\} = \atsurf{\trace_{\Gb}\Qb} - \atsurf{Q_{\xi\xi}} = - \atsurf{Q_{\xi\xi}}\).
We thus project \( \left\{ \atsurf{Q_{ij}} \right\} \) to \( \mathcal{Q}(\surf) \) with the orthogonal projection
\begin{align}\label{eq:qtensorprojection}
  \ProjQ:\tangent^{(2)}\surf &\rightarrow \mathcal{Q}(\surf)\formComma \quad
  \tb \mapsto \frac{1}{2}\left( \tb + \tb^{T} - (\trace_{\gb}\tb)\gb \right) \formComma
\end{align}
and define \( \qb\in\mathcal{Q}(\surf) \) by
\begin{align}\label{eq:SurfQExpansion}
  \qb &:= \ProjQ\left\{\atsurf{Q_{ij}}\right\}
           = \left\{\atsurf{Q_{ij}}\right\} + \frac{\beta}{2}\gb + \landau(h^{2})\formPeriod
\end{align}
For \( \Psib\in \mathcal{Q}_{0}(\surf_h)\) the tangential part is already a Q-tensor up to \( \landau(h^{2}) \).
Therefore we define $\psi_{ij}  := \atsurf{\Psi_{ij}} + \frac{1}{2}\atsurf{\psi_{\xi\xi}}g_{ij} = \atsurf{\Psi_{ij}} + \landau(h^{2})\formComma$
where \( \psib\in\mathcal{Q}(\surf) \).
With \eqref{eq:CoDerivativeThin}, \eqref{eq:CoDerivativeSurf}, \eqref{eq:ChristoffelExpansion}, \eqref{eq:AnchoringExpansion},
\eqref{eq:SurfQExpansion}, and the tensor shift $\betashift{\omega}{\qb} := \qb - \frac{\omega}{2}\beta\gb \formComma$
we can determine all covariant derivatives restricted to the surface by
\begin{align}
  \begin{aligned}
  \atsurf{Q_{\xi\xi;\xi}} &= \atsurf{\partial_{\xi}Q_{\xi\xi}} = \landau(h^{2}) \\
  \atsurf{Q_{i\xi ;\xi}} &= \atsurf{Q_{\xi i;\xi}}
              = \atsurf{\partial_{\xi}Q_{i \xi}} - \atsurf{\GGamma_{\xi i}^{K}Q_{K\xi}} = \landau(h^{2})\\
  \atsurf{Q_{\xi\xi;k}}
              &= \atsurf{\partial_{k}Q_{\xi\xi}} - 2\atsurf{\GGamma_{k \xi}^{L}Q_{L\xi}} = \landau(h^{2})\\
   \atsurf{Q_{i\xi ;k}} &= \atsurf{Q_{\xi i;k}}
              = \atsurf{\partial_{k}Q_{i \xi}} - \atsurf{\GGamma_{k i}^{l}Q_{l\xi}} - \atsurf{\GGamma_{k i}^{\xi}Q_{\xi\xi}} - \atsurf{\GGamma_{k \xi}^{l}Q_{il}}\\
              &= -\beta B_{ik} + \left( q_{il} - \frac{\beta}{2}g_{il} \right) \tensor{B}{^{l}_{k}} + \landau(h^{2})
               = \left[ \betashift{3}{\qb}\shapeOperator \right]_{ik} + \landau(h^{2}) \\
   \atsurf{Q_{ij;\xi}}
              &= - \frac{ L_{3}}{L_{1} + L_{6}\beta} \atsurf{Q_{i\xi;j}} + \landau(h^{2})
               = - \frac{ L_{3}}{L_{1} + L_{6}\beta}\left[ \betashift{3}{\qb}\shapeOperator \right]_{ik} + \landau(h^{2})\\
   \atsurf{Q_{ij;k}}
              &= \atsurf{\partial_{k}Q_{i j}} - \atsurf{\GGamma_{ki}^{l}Q_{lj}} - \atsurf{\GGamma_{ki}^{\xi}Q_{\xi j}}
                                              - \atsurf{\GGamma_{kj}^{l}Q_{il}} - \atsurf{\GGamma_{kj}^{\xi}Q_{i \xi}} \\
              &=\atsurf{\partial_{k}Q_{i j}} - \atsurf{\Gamma_{ki}^{l}Q_{lj}} - \atsurf{\Gamma_{kj}^{l}Q_{il}} + \landau(h^{2})
                = \left[ \betashift{1}{\qb} \right]_{ij|k} + \landau(h^{2})\\
              &= q_{ij|k} + \landau(h^{2}) \formPeriod
   \end{aligned}
\end{align}
Analogously, for the components of the covariant derivative \( \atsurf{\nabla\Psib} \), we obtain
$\atsurf{\Psi_{I\xi;\xi}} = \atsurf{\Psi_{\xi I;\xi}} = \atsurf{\Psi_{\xi\xi;I}} = \landau(h^{2})$,
$\atsurf{\Psi_{i\xi;k}} = \atsurf{\Psi_{\xi i ;k}} = \left[ \psib\shapeOperator \right]_{ik} + \landau(h^{2})$ and
$\atsurf{\Psi_{ij;k}} = \psi_{ij|k} + \landau(h^{2}) \formPeriod$
Note, in absence of natural boundary conditions for \( \Psib \), the covariant normal derivatives \( \atsurf{\Psi_{ij;\xi}} \) of the tangential components stay undetermined.
However, as we will see, the thin film limit of the \( L^{2} \)-gradient flow \eqref{eq:ELEqThin} does not depend on these derivatives. Adding up the three terms in \eqref{eq:ThinFilmEnergyIndices} with factors \( L_{1} \), \( L_{3} \), and \( L_{6} \),
factoring \( \nabla\Qb \) out, restricting to the surface
and considering \eqref{eq:RiseLowNormalIndex} and \eqref{eq:ContractionThinToSurf}, results in
\begin{align}
\begin{aligned}
 L_{1} \atsurf{\left\|\nabla\Qb\right\|^2_{\Gb}} &+ L_{3}\atsurf{\left\langle \nabla\Qb , \left( \nabla\Qb \right)^{T_{(2\,3)}} \right\rangle_{\Gb}}
                                          + L_{6} \atsurf{\left\langle \left( \nabla\Qb \right)\Qb , \nabla\Qb \right\rangle_{\Gb}} \\
      &= \atsurf{Q_{IJ;K}\left( L_{1}Q^{IJ;K} + L_{3}Q^{IK;J} + L_{6}Q^{KL}\tensor{Q}{^{IJ}_{;L}}\right)}\\
      &= \atsurf{Q_{ij;k}\left( L_{1}Q^{ij;k} + L_{3}Q^{ik;j} + L_{6}Q^{kl}\tensor{Q}{^{ij}_{;l}}\right)}\\
      &\quad  + \atsurf{Q_{\xi j;k}\left( L_{1}Q^{\xi j;k} + L_{3}Q^{\xi k;j} + L_{6}Q^{kl}\tensor{Q}{^{\xi j}_{;l}}\right)}\\
      &\quad  + \atsurf{Q_{i\xi ;k}\left( L_{1}Q^{i\xi ;k} + L_{3}Q^{ik;\xi } + L_{6}Q^{kl}\tensor{Q}{^{i\xi }_{;l}}\right)}\\
      &\quad  + \atsurf{Q_{ij;\xi }\left( L_{1}Q^{ij;\xi } + L_{3}Q^{i\xi ;j} + L_{6}Q^{\xi \xi }\tensor{Q}{^{ij}_{;\xi }}\right)}+ \landau(h^{2})\\
      &= \left(L_{1}-\frac{\beta}{2}L_{6}\right) \|\nabla\qb\|^2_{\gb} + L_{3}\left\langle \nabla\qb , \left( \nabla\qb \right)^{T_{(2\,3)}} \right\rangle_{\gb}
                                          + L_{6} \left\langle \left( \nabla\qb \right)\qb , \nabla\qb \right\rangle_{\gb} \\
      &\quad + \left( 2L_{1} - \frac{L_{3}^{2}}{L_{1}+L_{6}\beta} \right)\left\| \betashift{3}{\qb}\shapeOperator \right\|^{2}_{\gb}
             + L_{3}\trace_{\gb}\left( \betashift{3}{\qb}\shapeOperator \right)^{2}\\
      &\quad + 2L_{6} \left\langle \betashift{3}{\qb}\shapeOperator\betashift{1}{\qb}, \betashift{3}{\qb}\shapeOperator \right\rangle_{\gb}+ \landau(h^{2}) \formPeriod
\end{aligned}
\end{align}
With \( \trace\qb^{3} = 0 \) we obtain for the remaining terms
\begin{align}
  \atsurf{\left\| \Div\Qb \right\|^{2}_{\Gb}}
      &= \atsurf{\tensor{Q}{_{I}^{J}_{;J}}\tensor{Q}{^{IK}_{;K}}}
       = \atsurf{\tensor{Q}{_{i}^{j}_{;j}}\tensor{Q}{^{ik}_{;k}}} + \atsurf{\tensor{Q}{_{\xi}^{j}_{;j}}\tensor{Q}{_{\xi}^{k}_{;k}}} +\landau(h^{2}) \notag\\
      &= \left\| \Div\qb \right\|_{\gb}^{2} + \left( \trace_{\gb} \left(\betashift{3}{\qb}\shapeOperator\right) \right)^{2} +\landau(h^{2}) \\
  \atsurf{\trace_{\Gb}\Qb^{2}} &= \atsurf{Q_{IJ}Q^{JI}}
                             = \atsurf{Q_{ij}Q^{ji}} + \atsurf{(Q_{\xi\xi})^{2}} +\landau(h^{2})\notag\\
                             &= \trace_{\gb}\left( \qb - \frac{\beta}{2}\gb \right)^{2} + \beta^{2} +\landau(h^{2})
                             = \trace_{\gb}\qb^{2} + \frac{3}{2}\beta^{2} +\landau(h^{2})\\
  \atsurf{\trace_{\Gb}\Qb^{3}} &= \atsurf{Q_{IJ}Q^{JK}\tensor{Q}{_{K}^{I}}}
                                = \atsurf{Q_{ij}Q^{jk}\tensor{Q}{_{k}^{i}}} + \atsurf{(Q_{\xi\xi})^{3}} +\landau(h^{2})\notag\\
                               &= \trace_{\gb}\left( \qb - \frac{\beta}{2}\gb \right)^{3} + \beta^{3} +\landau(h^{2})
                                = \frac{3}{2}\beta\left( \frac{\beta^{2}}{2} - \trace_{\gb}\qb^{2} \right) +\landau(h^{2}) \label{eq:energysurfBulkQ3}\\
  \atsurf{\trace_{\Gb}\Qb^{4}} &= \frac{1}{2} \atsurf{\left(\trace_{\Gb}\Qb^{2}\right)^{2}}
                                = \trace_{\gb}\qb^{4} + \frac{3}{2}\beta^{2}\trace_{\gb}\qb^{2} + \frac{9}{8}\beta^{4} +\landau(h^{2})\formPeriod
\end{align}
Adding all these up, we can define \( \mathcal{F}^{\surf} := \mathcal{F}^{\surf}_\text{el} + \mathcal{F}^{\surf}_\text{bulk} \) by
\begin{align}\label{eq:energysurf1}
  \begin{aligned}
    \mathcal{F}^{\surf}_\text{el}[\qb] &:= \frac{1}{2}\int_\surf
                               \left(L_{1}-\frac{\beta}{2}L_{6}\right) \|\nabla\qb\|^2
                             + L_{2}\left\| \Div\qb \right\|^{2}
                             + L_{3}\left\langle \nabla\qb , \left( \nabla\qb \right)^{T_{(2\,3)}} \right\rangle
                             + L_{6} \left\langle \left( \nabla\qb \right)\qb , \nabla\qb \right\rangle\\
                 &\quad\quad + \left( 2L_{1} - \frac{L_{3}^{2}}{L_{1}+L_{6}\beta} \right)\left\| \betashift{3}{\qb}\shapeOperator \right\|^{2}
                             + L_{2}\left( \trace \left(\betashift{3}{\qb}\shapeOperator\right) \right)^{2}
                             + L_{3}\trace\left( \betashift{3}{\qb}\shapeOperator \right)^{2}\\
                 &\quad\quad + 2L_{6} \left\langle \betashift{3}{\qb}\shapeOperator\betashift{1}{\qb}, \betashift{3}{\qb}\shapeOperator \right\rangle \dS\\
  \mathcal{F}^{\surf}_\text{bulk}[\qb] &:= \int_\surf \frac{1}{2}(2a - 2b\beta + 3c\beta^2)\trace\qb^{2} + c \trace\qb^{4} + \frac{\beta^2}{8}(12a + 4b\beta + 9c\beta^2)\,\dS
  \end{aligned}
\end{align}
and by the rectangle rule and \eqref{eq:VolumeElementExpansion}, we obtain for \( h\rightarrow 0 \)
\begin{align}\label{eq:rectangleEnergie}
  \begin{aligned}
  \frac{1}{h}\mathcal{F}^{\surfh} &= \frac{1}{h}\int_{\surfh}F^{\surfh}\dV
              = \frac{1}{h}\int_{-\frac{h}{2}}^{\frac{h}{2}} \int_{\surf}  \left( 1- \xi\meanCurvature + \xi^{2}\gaussianCurvature \right)F^{\surfh}\dS\dxi
              = \int_{\surf} F^{\surf} \dS + \landau(h^{2})\\
              &= \mathcal{F}^{\surf} + \landau(h^{2})
                \longrightarrow \mathcal{F}^{\surf} \formPeriod
  \end{aligned}
\end{align}

Consequently, the energies \( \mathcal{F}^{\surf} \) and \( \mathcal{F}^{\surfh} \) are consistent \wrt\ the thickness \( h \).
To show a similar asymptotic behavior for the \( L^{2} \)-gradient flows, we investigate the first variation
\( \delta \mathcal{F}^{\surfh} = \delta\mathcal{F}^{\surfh}_{\text{el}} + \delta\mathcal{F}^{\surfh}_{\text{bulk}}\)
in \eqref{eq:ElFirstVar} and compare with the first variation
\( \delta \mathcal{F}^{\surf}= \delta\mathcal{F}^{\surf}_{\text{el}} + \delta\mathcal{F}^{\surf}_{\text{bulk}}\), where
\begin{align}
  \delta \mathcal{F}^{\surf}_{\text{el}}(\qb,\psib)
      &=  \int_\surf \left(L_{1}-\frac{\beta}{2}L_{6}\right) \left\langle \nabla\qb, \nabla\psib \right\rangle
                    + L_{2}\left\langle \Div\qb,\Div\psib \right\rangle
                    + L_{3} \left\langle \nabla\qb , \left( \nabla\psib \right)^{T_{(2\,3)}} \right\rangle \notag\\
      &\quad\quad   + L_{6}\left( \left\langle \left( \nabla\qb \right)\qb , \nabla\psib \right\rangle
                                + \frac{1}{2}\left\langle \left( \nabla\qb \right)\psib , \nabla\qb \right\rangle \right)
                    + \left( 2L_{1} - \frac{L_{3}^{2}}{L_{1}+L_{6}\beta} \right)\left\langle  \betashift{3}{\qb}\shapeOperator, \psib\shapeOperator \right\rangle\notag\\
      &\quad\quad   + L_{2}\left\langle \betashift{3}{\qb}, \shapeOperator \right\rangle \left\langle \shapeOperator, \psib \right\rangle
                    + L_{3}\left\langle \shapeOperator\betashift{3}{\qb}, \psib\shapeOperator \right\rangle\notag\\
      &\quad\quad   + L_{6}\left( 2\left\langle \betashift{3}{\qb}\shapeOperator\betashift{1}{\qb}, \psib\shapeOperator \right\rangle
                                 +\left\langle\shapeOperator \left( \betashift{3}{\qb} \right)^{2},  \psib\shapeOperator \right\rangle\right) \,\dS\\
  \delta \mathcal{F}^{\surf}_{\text{bulk}}(\qb,\psib)
      &=  \int_\surf (2a - 2b\beta + 3c\beta^2)\left\langle \qb, \psib \right\rangle
                    + 2c\trace\qb^{2}\left\langle \qb, \psib \right\rangle \,\dS \formPeriod
\end{align}
Proceeding as before we restrict the terms under the integral of \( \delta \mathcal{F}^{\surfh} \) in \eqref{eq:ElFirstVar} to the surface. For
\( \delta \mathcal{F}^{\surfh}_{\text{el}} \) we obtain
\begin{align}
  \Psi_{IJ;K}\Big( L_{1}Q^{IJ;K} &+ \atsurf{L_{3}Q^{IK;J} + L_{6}Q^{KL}\tensor{Q}{^{IJ}_{;L}} \Big)}\notag\\
                  &= \atsurf{\Psi_{ij;k}\left( L_{1}Q^{ij;k} + L_{3}Q^{ik;j} + L_{6}Q^{kl}\tensor{Q}{^{ij}_{;l}} \right)}\notag\\
                  &\quad  +\atsurf{\Psi_{i\xi;k}\left( 2L_{1}Q^{i\xi;k} + L_{3}\left( Q^{ik;\xi} + Q^{k\xi;i} \right) + 2L_{6}Q^{kl}\tensor{Q}{^{i\xi}_{;l}} \right)}
                          +\landau(h^{2})\notag\\
                  &= \psi_{ij|k}\left( \left( L_{1} - \frac{\beta}{2}L_{6} \right)q^{ij|k} + L_{3}q^{ik|j} + L_{6}q^{kl}\tensor{q}{^{ij}_{|l}} \right)\notag\\
                  &\quad  + \left[ \psib\shapeOperator \right]_{ik} \left(
                                                 \left( 2L_{1} - \frac{L_{3}^{2}}{L_{1}+L_{6}\beta} \right) \left[  \betashift{3}{\qb}\shapeOperator \right]^{ik}
                                                + L_{3}\left[  \betashift{3}{\qb}\shapeOperator \right]^{ki} \right)\notag\\
                  &\quad  +2L_{6} \left[ \psib\shapeOperator \right]_{ik} \left[  \betashift{3}{\qb}\shapeOperator\betashift{1}{\qb} \right]^{ik}  +\landau(h^{2})\\
  \atsurf{\tensor{\Psi}{_{I}^{J}_{;J}}\tensor{Q}{^{IK}_{;K}}}
                  &= \tensor{\psi}{_{i}^{j}_{|j}}\tensor{q}{^{ik}_{|k}}
                    +\tensor{\left[ \psib\shapeOperator \right]}{^{j}_{j}} \tensor{\left[ \betashift{3}{\qb}\shapeOperator \right]}{^{k}_{k}} +\landau(h^{2}) \\
  \frac{1}{2}\atsurf{\Psi_{IJ}\tensor{Q}{_{KL}^{;I}}Q^{KL;J}}
                  &= \frac{1}{2}\atsurf{\Psi_{ij}\left( \tensor{Q}{_{kl}^{;i}}Q^{kl;j} + 2\tensor{Q}{_{k\xi}^{;i}}Q^{k\xi;j} \right)} +\landau(h^{2})\notag\\
                  &= \frac{1}{2}\psi_{ij}\tensor{q}{_{kl}^{|i}}q^{kl|j} + \psi_{ij}\left[ \shapeOperator(\betashift{3}{\qb})^{2}\shapeOperator  \right]^{ij} +\landau(h^{2})
\end{align}
and for \( \delta \mathcal{F}^{\surfh}_{\text{bulk}} \)
\begin{align}
  2\atsurf{\left(a+cQ_{KL}Q^{KL}\right)Q^{IJ}\Psi_{IJ}}
      &= \left( 2a + 2cq_{kl}q^{kl} + 3c\beta^{2} \right)\left( q^{ij}\psi_{ij} - \frac{\beta}{2}\tensor{\psi}{^{i}_{i}} \right)  +\landau(h^{2}) \notag\\
      &= \left( 2a + 3c\beta^{2} \right) \left\langle \qb, \psib \right\rangle + 2c\trace\qb^{2}\left\langle \qb, \psib \right\rangle +\landau(h^{2})\\
  2b\atsurf{Q^{IK}\tensor{Q}{_{K}^{J}}\Psi_{IJ}}
      &= 2b\atsurf{Q^{ik}\tensor{Q}{_{k}^{j}}\Psi_{ij}} +\landau(h^{2})
      = 2b\left( \left[ \qb^{2} \right]^{ij} - \beta q^{ij} + \frac{\beta^{2}}{4}g^{ij} \right)\psi_{ij} +\landau(h^{2}) \notag\\
      &= -2b\beta q^{ij}\psi_{ij} + b\left( \trace\qb^{2} + \frac{\beta^{2}}{2} \right)\tensor{\psi}{^{i}_{i}}+\landau(h^{2})\notag\\
      &= -2b \left\langle \qb, \psib \right\rangle +\landau(h^{2}) \formComma
\end{align}
where we used \autoref{lem:QuadraticIdenentities}, \ie, \( 2\qb^{2}=(\trace\qb^{2})\gb \), particularly.
In summary, we see
that
\( \atsurf{\left\langle \nabla_{L^{2}}\mathcal{F}^{\surf_h}, \Psib \right\rangle} =  \left\langle \nabla_{L^{2}}\mathcal{F}^{\surf}, \psib \right\rangle + \landau(h^{2}) \)
is valid.
Moreover, as \( \partial_{t}\gb = 0 \) for a stationary surface, we obtain \( \left[ \atsurf{\partial_{t}\Qb} \right]_{ij} = \left[ \partial_{t}\qb \right]_{ij} + \landau(h^{2}) \).
Finally, as in \eqref{eq:rectangleEnergie}, we argue with the rectangle rule in normal direction and observe
\begin{align}\label{eq:FirstVarThinVSSurf}
  \frac{1}{h}\int_{\surfh}\left\langle  \nabla_{L{2}}\mathcal{F}^{\surf_h} + \partial_{t}\Qb, \Psib \right\rangle\dV
      &=  \int_{\surf}\left\langle  \nabla_{L^{2}}\mathcal{F}^{\surf} + \partial_{t}\qb, \psib \right\rangle\dS + \landau(h^{2})\formPeriod
\end{align}

\subsubsection{Surface energy}\label{sec:surfaceEnergy}
To have a better distinction between extrinsic terms, \ie, \( \left\langle \shapeOperator, \qb \right\rangle \),
and terms depending only on scalar curvatures \( \meanCurvature \) and \( \gaussianCurvature \) in the surface energy \eqref{eq:energysurf1},
we use \autoref{lem:QuadraticIdenentities} and obtain the substitutions
\begin{align}
  \left( \trace\left( \betashift{3}{\qb}\shapeOperator \right) \right)^{2}
      &= \left\langle  \shapeOperator, \qb  \right\rangle^{2} - 3\beta\meanCurvature \left\langle  \shapeOperator, \qb  \right\rangle +\frac{9}{4}\beta^{2}\meanCurvature^{2}\\
  \trace\left( \betashift{3}{\qb}\shapeOperator \right)^{2}
      &= \left\langle  \shapeOperator, \qb  \right\rangle^{2} + \gaussianCurvature\trace\qb^{2}
          -  3\beta\meanCurvature \left\langle  \shapeOperator, \qb  \right\rangle
          + \frac{9}{4}\beta^{2}\left( \meanCurvature^{2} - 2\gaussianCurvature \right)\\
  \left\| \betashift{3}{\qb}\shapeOperator \right\|^{2}
      &= \frac{1}{2}\left(  \meanCurvature^{2} - 2\gaussianCurvature \right)\trace\qb^{2}
         -  3\beta\meanCurvature \left\langle  \shapeOperator, \qb  \right\rangle
          + \frac{9}{4}\beta^{2}\left( \meanCurvature^{2} - 2\gaussianCurvature \right)\\
   \left\langle \betashift{3}{\qb}\shapeOperator \betashift{1}{\qb} , \betashift{3}{\qb}\shapeOperator\right\rangle
       &= \frac{1}{2}\meanCurvature\trace\qb^{2}\left\langle  \shapeOperator, \qb  \right\rangle
        -\beta\left( 3 \left\langle  \shapeOperator, \qb  \right\rangle^{2} + \frac{1}{4}\left( \meanCurvature^{2} + 10\gaussianCurvature \right)\right)\trace\qb^{2}\notag\\
      &\quad  + \frac{15}{4}\meanCurvature\beta^{2}\left\langle  \shapeOperator, \qb  \right\rangle
        - \frac{9}{8}\beta^{3}\left( \meanCurvature^{2} - 2\gaussianCurvature \right)
\end{align}
at the surface \( \surf \).
Terms with invariant measurement of the gradient \( \nabla\qb \) differ only in zero order quantities for a closed surface, see \autoref{lem:derivationContractionsSurf}.
Adding all these up, we obtain \eqref{eq:SurfaceEnergy} and therefore in index notation
\begin{align}
  \begin{aligned}
    \mathcal{F}^{\surf}_\text{el}[\qb]
         &=\frac{1}{2}\int_\surf L'_{1} q_{ij|k}q^{ij|k} + L_{6} q^{kl}q_{ij|k}\tensor{q}{^{ij}_{|l}} \\
                         &\quad\quad + M_{1} q_{ij}q^{ij} + M_{2} B^{ij}B^{kl}q_{ij}q_{kl} + M_{3} B^{ij}q_{ij}q^{kl}q_{kl} + M_{4}B^{ij}q_{ij}
                                + C_{0}\,\dS \formComma\\
    \mathcal{F}^{\surf}_\text{bulk}[\qb] &:= \int_\surf a' q_{ij}q^{ij} + c q_{ij}q^{jk}q_{kl}q^{li} + C_{1}\,\dS \formComma\\
  \end{aligned}
\end{align}
with coefficient functions
\begin{align}\label{eq:parameterfunctions}
    \begin{aligned}
    L'_{1} &:= L_{1} + \frac{1}{2}\left( L_{2} + L_{3} - L_{6}\beta \right)\formComma\\
    M_{2}  &:= L_{2} + L_{3} - 6L_{6}\beta\formComma\\
    M_{3}  &:= L_{6}\meanCurvature\formComma\\
    M_{1}  &:= \frac{1}{2}\left( -L_{6}\left( \meanCurvature^{2}+10\gaussianCurvature \right)\beta
                                                                  +\left( 2L_{1} - \frac{L_{3}^{2}}{L_{1}+L_{6}\beta}\right)\left( \meanCurvature^{2}-2\gaussianCurvature \right)
                                                                  +\left( L_{2} + L_{3} \right)\gaussianCurvature \right)\formComma\\
    M_{4}  &:= -3\left( 2L_{1} + L_{2} + L_{3} - \frac{5}{2}L_{6}\beta - \frac{L_{3}^{2}}{L_{1}+L_{6}\beta} \right)\beta\meanCurvature\formComma\\
    C_{0}  &:= \frac{9}{4}\left( \left( 2L_{1} + L_{3} - L_{6}\beta - \frac{L_{3}^{2}}{L_{1}+L_{6}\beta} \right)\left( \meanCurvature^{2}-2\gaussianCurvature \right)
                                                                      +L_{2}\meanCurvature^{2} \right)\beta^{2}\formComma \\
    a'     &:= \frac{1}{2}(2a - 2b\beta + 3c\beta^2) \quad \quad \text{and} \\
    C_{1}  &:= \frac{\beta^2}{8}(12a + 4b\beta + 9c\beta^2) \formPeriod
    \end{aligned}
\end{align}

\subsubsection{Surface equation of motion}\label{sec:surfaceEquationOfMotion}
To obtain the strong form of the surface \( L^{2} \)-gradient flow $\partial_{t}\qb  = -\nabla_{L^{2}}\mathcal{F}^{\surf}$ we have to ensure
\begin{align}
 \int_{\surf}\left\langle \partial_{t}\qb, \psib \right\rangle \dS
                                            &= -\int_{\surf}\left\langle \nabla_{L^{2}}\mathcal{F}^{\surf}, \psib \right\rangle \dS, \quad  \forall\psib\in\mathcal{Q}(\surf),
\end{align}
\wrt\ the \( L^{2} \) inner product over the space of Q-tensors and thus \( \nabla_{L^{2}}\mathcal{F}^{\surf} \in \mathcal{Q}(\surf) \).
While for the first variations $\delta$ \wrt\ $\qb$ in direction $\psi$
\begin{align}
  \frac{1}{2}\delta\int_{\surf}\left\| \nabla\qb \right\|^{2}\,\dS
            &= \int_{\surf} \left\langle -\Div\nabla\qb , \psib  \right\rangle\,\dS\formComma\\
 \frac{1}{2}\delta\int_{\surf} \trace\qb^{2} \,\dS
            &= \int_{\surf} \left\langle \qb , \psib  \right\rangle\,\dS\formComma\\
 \frac{1}{2}\delta\int_{\surf} \trace\qb^{4} \,\dS
            &= \int_{\surf} \left\langle \left( \trace\qb^{2} \right)\qb , \psib  \right\rangle\,\dS\formComma
\end{align}
the left argument of the inner product is already in \( \mathcal{Q}(\surf) \),
we have to apply \( \ProjQ \) defined in \eqref{eq:qtensorprojection} for the remaining terms, \ie,
\begin{align}
   \frac{1}{2}\delta\int_{\surf} \left\langle \left( \nabla\qb \right)\qb , \nabla\qb \right\rangle \,\dS
      &= \int_{\surf}\left( -\left( q_{ij|k}q^{kl} \right)_{|l} + \frac{1}{2}q_{kl|i}\tensor{q}{^{kl}_{|j}} \right)\psi^{ij} \,\dS\notag\\
      &= \int_{\surf}\left( - q_{ij|k|l}q^{kl} - q_{ij|k}\tensor{q}{^{kl}_{|l}}
                            + \frac{1}{2}\left[\ProjQ\left\{q_{kl|i}\tensor{q}{^{kl}_{|j}}\right\}\right]_{ij} \right)\psi^{ij}\,\dS\notag\\
      &= \int_{\surf}\left( - q_{ij|k|l}q^{kl} - q_{ij|k}\tensor{q}{^{kl}_{|l}}
                            +  \frac{1}{2}q_{kl|i}\tensor{q}{^{kl}_{|j}} -  \frac{1}{4}q_{kl|m}q^{kl|m}g_{ij}  \right)\psi^{ij} \,\dS\\
      &= \int_{\surf} \left\langle \left( -\nabla\nabla\qb \right):\qb - \left( \nabla\qb \right)\cdot\Div\qb
                      +\frac{1}{2}\left( \nabla\qb \right)^{T_{(1\,3)}}:\nabla\qb - \frac{1}{4}\left\| \nabla\qb \right\|^{2}\gb, \psib \right\rangle \,\dS\formComma\notag\\
  \frac{1}{2}\delta\int_{\surf} \left\langle \shapeOperator,\qb \right\rangle^{2} \,\dS
      &= \int_{\surf} \left\langle \left\langle \shapeOperator,\qb \right\rangle \shapeOperator , \psib \right\rangle \,\dS
       = \int_{\surf} \left\langle \left\langle \shapeOperator,\qb \right\rangle \ProjQ\shapeOperator , \psib \right\rangle \,\dS \notag\\
      &= \int_{\surf} \left\langle \left\langle \shapeOperator,\qb \right\rangle \left(\shapeOperator - \frac{1}{2}\meanCurvature\gb\right) , \psib \right\rangle \,\dS\formComma\\
  \frac{1}{2}\delta\int_{\surf} \trace\qb^{2}\left\langle \shapeOperator,\qb \right\rangle \,\dS
      &= \int_{\surf} \left\langle \frac{1}{2}\trace\qb^{2}\shapeOperator + \left\langle \shapeOperator,\qb \right\rangle\qb , \psib \right\rangle\,\dS
       = \int_{\surf} \left\langle \frac{1}{2}\trace\qb^{2}\ProjQ\shapeOperator + \left\langle \shapeOperator,\qb \right\rangle\qb , \psib \right\rangle\,\dS \notag\\
      &= \int_{\surf} \left\langle \frac{1}{2}\trace\qb^{2}\left(\shapeOperator - \frac{1}{2}\meanCurvature\gb\right)
                                    + \left\langle \shapeOperator,\qb \right\rangle\qb , \psib \right\rangle\,\dS\formComma \\
  \frac{1}{2}\delta\int_{\surf} \left\langle \shapeOperator,\qb \right\rangle \,\dS
      &= \int_{\surf}\left\langle \frac{1}{2}\shapeOperator, \psib \right\rangle \,\dS
       = \int_{\surf}\left\langle \frac{1}{2}\ProjQ\shapeOperator, \psib \right\rangle \,\dS \notag\\
      &= \int_{\surf}\left\langle \frac{1}{2}\left(\shapeOperator - \frac{1}{2}\meanCurvature\gb\right), \psib \right\rangle \,\dS\formPeriod
\end{align}
Finally, with \( \left[ \Delta^{dG}\qb \right]_{ij} := \tensor{q}{_{ij}^{|k}_{|k}} \), the div-Grad (Bochner) Laplace operator,
we get the equation of motion \eqref{eq:evoliEquation}, which reads in index notation
\begin{align}
 \partial_{t}q_{ij}  &= L'_{1} \tensor{q}{_{ij}^{|k}_{|k}}
     + L_{6}\left(  - q_{ij|k|l}q^{kl} - q_{ij|k}\tensor{q}{^{kl}_{|l}}
                            +  \frac{1}{2}q_{kl|i}\tensor{q}{^{kl}_{|j}} -  \frac{1}{4}q_{kl|m}q^{kl|m}g_{ij}  \right)\\
    &\quad  - \left( M_{1} + M_{3}B_{kl}q^{kl}  + 2a' +2c q_{kl}q^{kl} \right)q_{ij}
      - \left( M_{2} B_{kl}q^{kl} + \frac{M_{3}}{2}q_{kl}q^{kl} + \frac{M_{4}}{2} \right)\left(B_{ij} - \frac{1}{2}\meanCurvature g_{ij}\right)\notag \formPeriod
\end{align}

After establishing weak consistences for the energies and the \( L^{2} \)-gradient flows in the thin film and at the surface in \eqref{eq:rectangleEnergie} and \eqref{eq:FirstVarThinVSSurf}, we also have pointwise consistence for the evolution equation in the Q-tensor space restricted to the surface for sufficient regularity, \ie,
\begin{align}
  \left\| \ProjQ\left[ \ProjSurf\atsurf{\left( \partial_{t}\Qb + \nabla_{L^{2}}\mathcal{F}^{\surfh}\left[ \Qb \right] \right)}\ProjSurf\right]
      -\left(\partial_{t}\qb + \nabla_{L^{2}}\mathcal{F}^{\surf}\left[ \qb \right]\right) \right\|_{\gb} &= \landau(h^{2})
\end{align}
\wrt\ boundary conditions for \( \Qb \) at \( \partial\surfh \) and initial condition
\( \qb|_{t=0} = \ProjSurf\Qb|_{(\surf,t=0)}\ProjSurf + \normal\Qb|_{(\surf,t=0)}\normal \).
This means, the order of performing the limit \( h\rightarrow 0 \) and formulating the local dynamic equation, \wrt\ \( \nabla_{L^{2}} \) flow, does not matter,
\ie, the diagram
\begin{equation}
  \begin{tikzcd}
    \mathcal{F}^{\surfh} \arrow{r}{h\rightarrow 0 }[swap]{\frac{1}{h}} \arrow[mapsto]{d}[swap]{\nabla_{L^{2}}\text{flow}} & \mathcal{F}^{\surf} \arrow[mapsto]{d}{\nabla_{L^{2}}\text{flow}}\\
        \partial_{t}\Qb = -\nabla_{L^{2}}\mathcal{F}^{\surfh}\left[ \Qb \right]  \arrow{r}{h\rightarrow 0 }[swap]{\ProjQ, \ProjSurf}
                & \partial_{t}\qb = -\nabla_{L^{2}}\mathcal{F}^{\surf}\left[ \qb \right]
  \end{tikzcd}
\end{equation}
commutes.

\section{Discussion}
We now discuss similarities and differences between the thin film and surface Landau-de Gennes energy and their physical implication. Besides the terms containing the extrinsic quantity \( \shapeOperator \) and its scalar valued invariants, the surface Q-tensor energy \eqref{eq:SurfaceEnergy} is similar to the thin film Q-tensor energy \eqref{eq:ThinFilmEnergy}. While we have three scalar invariants for the gradient \( \nabla_{\Gb}\Qb \) in the thin film controlled by \( L_{1} \), \( L_{2} \), and \( L_{3} \),
at the surface we need only one for \( \nabla_{\gb}\qb \) to formulate the distortion of \( \qb \), see \autoref{lem:derivationContractionsSurf}.
This behavior seems to be a consequence of reducing the degree of freedoms of Q-tensors.
Particularly, \( \mathcal{Q}(\surf_h) \) is a five dimensional function-vector space, while \( \mathcal{Q}(\surf) \) is only a two dimensional function-vector space with improper rotation endomorphisms in the tangential bundle
as basis tensors.
Moreover, at the surface we only consider the trace of even powers of \( \qb \) for the bulk energy as
for \( r\ge 0 \) it holds
\begin{align}
  \trace\qb^{2(r+1)} &= \left\langle \left( \qb^{2} \right)^{r+1}, \gb \right\rangle
                      = 2^{-(r+1)}\left( \trace\qb^{2} \right)^{r+1}\left\| \gb \right\|^{2} = 2^{-r}\left( \trace\qb^{2} \right)^{r+1}\\
  \trace\qb^{2r+1} &= \left\langle \left( \qb^{2} \right)^{r}, \qb \right\rangle
                      = 2^{-r}\left( \trace\qb^{2} \right)^{r} \trace{\qb} = 0, \label{eq:oddTensorTrace}
\end{align}
see \autoref{lem:QuadraticIdenentities}. This has several consequences. In principle it leads to a change in phase transition type, as coexistence between a nematic and an isotropic phase is not possible without the $\trace\qb^3$ term. However, as we will see, our model still allows coexistence. We first show that we can preserve the phase diagram of the thin film bulk energy. To limit complexity we have considered $\normal \Qb \normal=\beta$ to be constant. Similar assumptions have been made in \cite{Napoli2012b, Kralj2011}. Our approach chooses $\beta$ such that surface and thin film formulation of bulk energy match. For $\beta=-\frac{1}{3} S^*$, where $S^*=\frac{1}{4c}(-b+\sqrt{b^2-24ac})$ indeed the minima of $\mathcal{F}^{\surf_h}_\text{bulk}$ and $\mathcal{F}^{\surf}_\text{bulk}$ are equal and are achieved for $S=S^*$, with $S = S_1 = S_2$ or $S = S_1$ if $S_2 = 0$ or $S = S_2$ if $S_1 = 0$. The reconstructed thin film Q-tensor $\Qb = \qb-\frac{\beta}{2}\ProjSurf + \beta\normal\otimes\normal$ is uniaxial with eigenvalues $[\frac{2}{3}S,-\frac{1}{3}S,-\frac{1}{3}S]$. \autoref{fig:phaseDiagramBulk}  shows the phase diagram. Contrary to the modeling via degenerate states with $\beta=0$, see, \eg, \cite{Kralj2011}, the phase diagram of the bulk energy is preserved for $\beta = - \frac{1}{3} S^*$.

\begin{figure}[ht]
\begin{center}
\includegraphics[width=.45\linewidth]{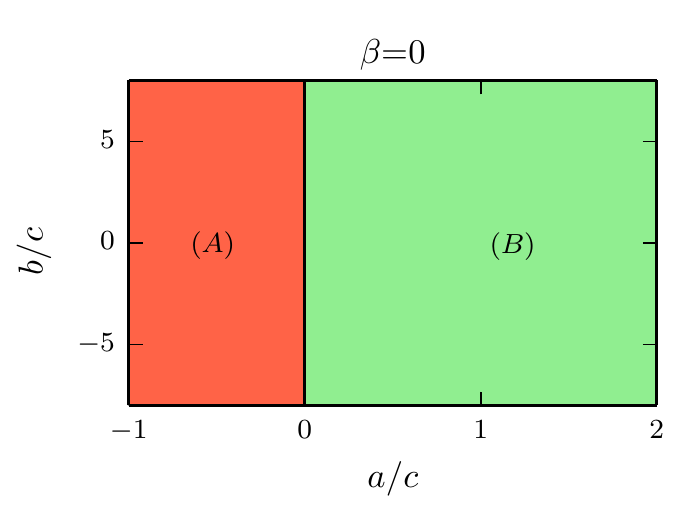}
\includegraphics[width=.45\linewidth]{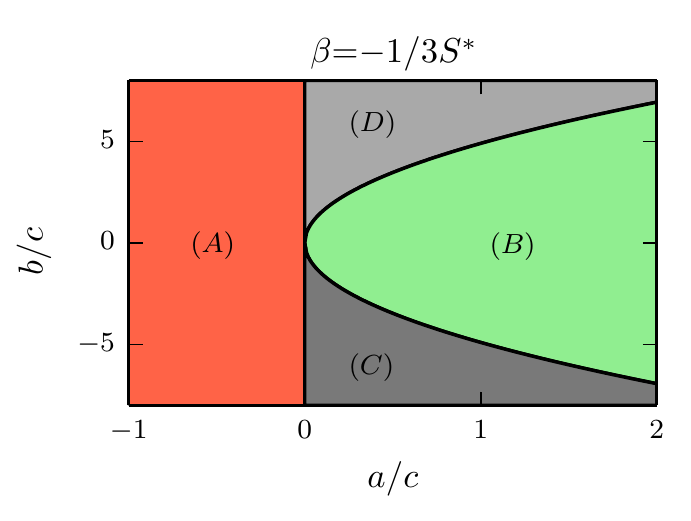}
\end{center}
\caption{\textbf{Phase diagram of bulk energy vs choice of $\beta\,$:} (colors online) (left) Double-well potential phase diagram for $\beta=0$ exhibiting two domains enabling the existence of (A): stable nematic ordering $S^*\neq 0$ or (B): stable isotropic ordering $S^*=0$. (right) Phase diagram for $\beta = -\frac{1}{3} S^*$ enabling additional stable phases discriminating between (C): only tangential nematic ordering is stable, $S^*>0$ or (D): only normal nematic ordering is stable, $S^*<0$. As we are interested only in tangential anchoring (D) is not within the scope of this paper.}
\label{fig:phaseDiagramBulk}
\end{figure}

With the emergence of defects the assumption $\beta = const$ becomes questionable and a more precise modeling would require to treat $\beta$ as a degree of freedom. However, this would lead to an excessive amount of additional coupling terms in the elastic energy and thus makes the complexity of the model infeasible. A detailed derivation and interpretation of the additional terms thus remains an open question.

\begin{figure}[ht]
\begin{center}
\includegraphics[width=.95\linewidth]{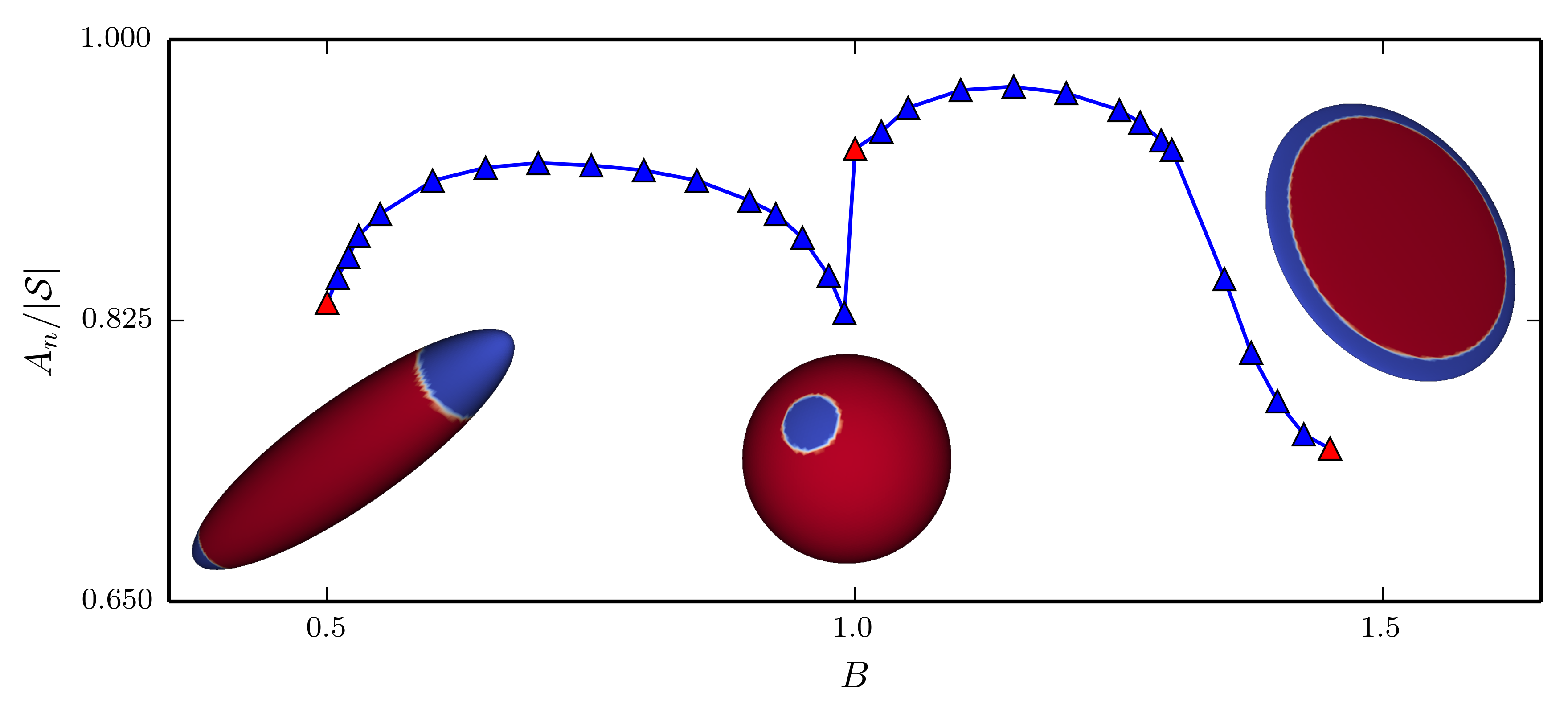}
\end{center}
\caption{\textbf{Curvature controls isotropic-nematic phase coexistence:} (colors online) Relative area of the nematic phase $A_n / |\surf|$ as a function of the geometry of the ellipsoid, parameterized by its axis $B$. For prolates ($B < 1.0$) the isotropic phases are located at the high curvature regions at the poles. They increase with increasing curvature for $B \lesssim 0.6$. For oblates ($B > 1.0$) the isotropic phase is located at the high curvature region along the rim. It increases with increasing curvature for $B \gtrsim 1.2$. The non-monotone behavior in between results from a rearrangement of two regions on a prolate to four regions on an oblate, which merge for larger $B$. The inlets show realizations with red corresponding to the nematic and blue to the isotropic phase. The corresponding shape parameters are highlighted with red triangle markers. To distinguish between both phases a threshold of $10\%$ of the expected norm of $\qb$ is used. The model parameters are the same as in \autoref{fig:defectDynamics}, except $\omega = 2.5$ to highlight the behavior already for moderate curvatures.}
\label{fig:shapeVariation}
\end{figure}

Considering the elastic energy, the surface model provides a set of new terms consisting of combinations of $\trace\qb^2$ and $\left\langle \shapeOperator,\qb \right\rangle $. These terms interact with the double-well potential $a'\trace\qb^2 + c \trace \qb^4$ of the surface bulk energy. By this interaction the bulk potential can be deformed locally, as e.g. $M_1 \trace\qb^2$, depends on geometric properties $M_1=M_1(\meanCurvature,\gaussianCurvature)$. So, while the bulk potential itself inhibits isotropic-nematic phase coexistence, a global phase coexistence can emerge on surfaces by local variance of geometric properties, see \autoref{fig:shapeVariation}.

\begin{figure}[ht]
\begin{center}
\includegraphics[width=.99\linewidth]{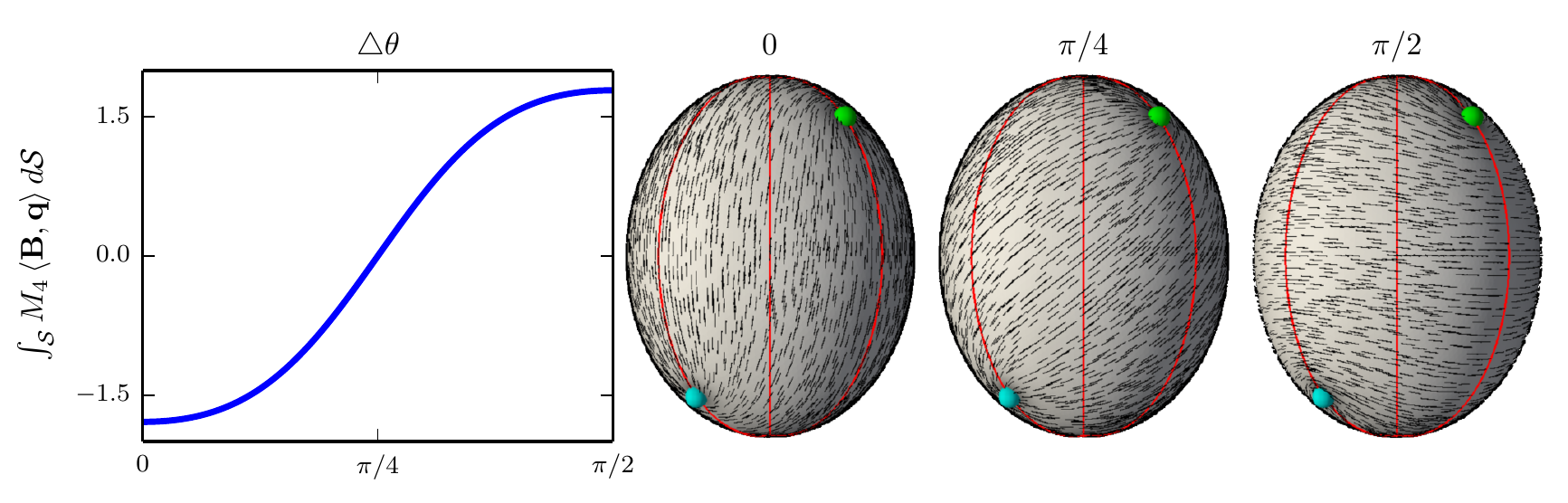}
\end{center}
\caption{\textbf{$\left\langle \shapeOperator,\qb \right\rangle$ term removes rotational invariance of elastic energy:} (colors online) (from left to right) Elastic energy contribution of $\int_{\surf}M_4\left\langle \shapeOperator,\qb \right\rangle \dS$ under rotation $\triangle\theta$ of Q-tensor field $\qb$ . Energetic minimum at $\triangle\theta=0$ with director parallel to lines of minimal curvature(marked in red), increased energy at intermediate state at $\triangle\theta = \pi/4$ and maximal energy for director orthogonal to lines of minimal curvature at $\triangle\theta = \pi/2$. Energy contributions of $L_1'$ and $M_1$ are invariant under rotation and therefore constant. The model parameters are the same as in \autoref{fig:defectDynamics}.}
\label{fig:BqTerm}
\end{figure}

The term  $\left\langle \shapeOperator,\qb \right\rangle $ imposes restrictions on energetic favorable ordering. This term can be expressed in terms of principal director $\Pb$ of $\qb$ by $\left\langle \shapeOperator,\qb \right\rangle =\Pb\shapeOperator\Pb -\frac{1}{2} \meanCurvature \|\Pb\|$ illustrating a geometric forcing towards the ordering along lines of minimal curvature. Such forcing does eliminate the rotational invariance of the four $+\frac{1}{2}$ defect configuration on an ellipsoid as demonstrated in \autoref{fig:BqTerm}.
The same effect has also been observed in surface Frank-Oseen modell for surface polar liquid crystals \cite{Nestler2017}.

Combining these effects provides a wide range of intriguing mechanisms coupling geometry and ordering with significant impacts on minimum energy states and dynamics. A more detailed elaboration of these interactions as well as a detailed description of the used numerical approach will be given elsewhere.

As a complementary result, we point out that the surface model for degenerate states in \cite{Kralj2011} can be reproduced by our model by choosing $\beta=0$, $k=L_{1}= \frac{1}{\sqrt{2}}L_{2}=-\frac{1}{\sqrt{2}}L_{3}, k_{24} = -\sqrt{2}k$ and defining $2a=A$, $2c=C$. A one-to-one comparison with the models derived in \cite{Napoli2012b,Golovaty2017,Canevari2017} is more complicated, as in contrast to our approach, which only uses the Levi-Civita connections $\nabla$, other surface derivatives are introduced in \cite{Napoli2012b,Golovaty2017,Canevari2017}, which make these models depending on the chosen coordinate system. A detailed comparison of numerical simulations might allow to point out similarities and differences.

\section{Conclusion}
We have rigorously derived a surface Q-tensor model by performing the thin film limit. Instead of making assumptions on the Q-tensor field in the thin film we have prescribed a set of boundary conditions for the thin film. By requiring the normal components of $\Qb$ to be compatible with the minimum of the bulk energy we were able to transfer main features of the thin film model, like uniaxiality or parameter-phase space, to the surface model. Nonetheless, these features break down in areas of defects. It still remains an open question how to treat defect areas properly in surface Q-tensor models.

The proposed approach to derive thin film limits is general and can also be used for other tensorial problems, e.g. in elasticity. Note that for deriving thin film limits containing higher order derivatives, also higher order expansions for thin film metric quantities are needed, \eg\ \( \GGamma_{ij}^{k} = \Gamma_{ij}^{k} + \xi\tensor{\Theta}{_{ij}^{k}} + \landau(\xi^{2}) \) with
\( \tensor{\Theta}{_{ij}^{k}} := \tensor{B}{_{i}^{k}_{|j}} + \tensor{B}{_{j}^{k}_{|i}} - \tensor{B}{_{ij}^{|k}} \) for the  pure tangential components of Christoffel symbols to express
second order covariant derivatives like the Laplace operator \( \Delta \) in the thin film. Our analysis also indicates that the surface evolution equation can be derived directly without a detour of a global energy minimization problem.
However, there is no general theory regarding sufficient prerequisites of this analysis, and we can not ensure, that, for example, every well posed tensorial thin film problem results in a well posed tensorial surface problem.

Even with the made approximations in the modeling approach the numerical results provide new insights on the tight coupling of topology, geometry, and energetic minimal states as well as dynamics. In a next step the derived coupling terms should be investigated systematically and the model should be validated versus experimental data.
Various extensions of the proposed model, like coupling to hydrodynamics and/or activity open up a wide array of possible physical applications in material science or biophysics. For recent work on hydrodynamics on surfaces we refer to \cite{Nitschke2012,Reuther2015,Nitschke2017,Reuther2017}. Also investigations on energy minimization and dynamics on moving domains seem now feasible. However, to deal with these problems numerically requires a more detailed investigation of the regularity. In contrast to our assumption for the tensor fields to be sufficiently smooth, which was made for simplicity, tensorial Sobolev spaces should be investigated, see e.g. \cite{Segatti2016}.

{\bf Acknowledgements}

HL and AV acknowledge financial support from DFG through Lo481/20 and Vo899/19, respectively. We further acknowledge computing resources provided by JSC under grant HDR06 and ZIH/TU Dresden.

\appendix
\section{Appendix}

  \begin{lem}\label{lem:derivationContractionsSurf}
    For all surface q-Tensors \( \qb\in\mathcal{Q}(\surf) \) holds
    \begin{align}
      \int_{\surf}\left\| \Div\qb \right\|^{2}\dS
              &= \int_{\surf} \frac{1}{2}\left\| \nabla\qb \right\|^{2} + \gaussianCurvature\trace\qb^{2}\dS \formComma\\
      \int_{\surf}\left\langle \nabla\qb,  \left( \nabla\qb \right)^{T_{(2\,3)}}\right\rangle\dS
              &= \int_{\surf} \frac{1}{2}\left\| \nabla\qb \right\|^{2} - \gaussianCurvature\trace\qb^{2}\dS \formPeriod
    \end{align}
  \end{lem}
  \begin{proof}
    With the surface Levi-Civita tensor \( \Eb\cong\dS \) defined by
    \begin{align}\label{eq:LeviCivitaTensorSurf}
      E_{ij} &:= \dS(\partial_{i}\xb, \partial_{j}\xb) = \sqrt{\det\gb}\,\varepsilon_{ij} \formComma
    \end{align}
    with Levi-Civita symbols \( \varepsilon_{ij} \), we use the 2-tensor curl
    \begin{align}
      \left[ \Rot\qb \right]_{i} &:= \left[ -\nabla\qb : \Eb \right]_{i} = -E_{jk}\tensor{q}{_{i}^{jk}}
    \end{align}
    and observe
    \begin{align}
      \left[ -\Eb\cdot\Rot\qb \right]_{i}
            &= E_{il}E_{jk}q^{lj|k}
             = \left( g_{ij}g_{lk} - g_{ik}g_{lj} \right)q^{lj|k}
             = \tensor{q}{^{l}_{i|l}} - \tensor{q}{_{j}^{j}_{|i}}
             = \tensor{q}{_{i}^{l}_{|l}}
             = \left[ \Div\qb \right]_{i} \formPeriod
    \end{align}
    Moreover, in this case, \( -\Eb\cdot \) is isomorph to the  Hodge-star operator \( * \) on differential 1-forms
    and therefore it can be seen as a length preserving pointwise counterclock quarter turn,
    that is why \( \left\| \Rot\qb \right\| = \left\| -\Eb\cdot\Rot\qb \right\| = \left\| \Div\qb \right\| \) holds for the norm.
    We remark, that \( \Eb \in \tangent^{(2)}\surf \) is compatible with \( \nabla \) and
    hence, we calculate
    \begin{align}
      \begin{aligned}
      \int_{\surf}\left\| \Div\qb \right\|^{2}\dS
          &= \frac{1}{2}\int_{\surf} \left\| \Div\qb \right\|^{2} + \left\| \Rot\qb \right\|^{2}\dS
          = -\frac{1}{2}\int_{\surf} \left( \tensor{q}{_{i}^{k}_{|k|l}} + E_{kj}E_{lm}\tensor{q}{_{i}^{k|j|m}} \right) q^{il}\dS \\
          &= -\frac{1}{2}\int_{\surf} \left( \tensor{q}{_{i}^{k}_{|k|l}} + \tensor{q}{_{il}^{|j}_{|j}} - \tensor{q}{_{i}^{k}_{|l|k}} \right) q^{il}\dS\formPeriod
          \end{aligned}
    \end{align}
    The Riemannian curvature tensor has only one independent component on surfaces and is given by
    \( \boldsymbol{R} = \gaussianCurvature\Eb\otimes\Eb \in \tangent^{(4)}\surf\).
    Hence, for changing the order of covariant derivatives, holds
    \begin{align}
      \begin{aligned}
      \tensor{q}{_{i}^{k}_{|k|l}} -\tensor{q}{_{i}^{k}_{|l|k}}
          &= \tensor{R}{^{j}_{ikl}}\tensor{q}{_{j}^{k}} - \tensor{R}{^{k}_{jkl}}\tensor{q}{_{i}^{j}}
          = \gaussianCurvature \left( \left( \tensor{\delta}{^{j}_{k}}g_{il} -  \tensor{\delta}{^{j}_{l}}g_{ik} \right)\tensor{q}{_{j}^{k}}
                                     - \left( \tensor{\delta}{^{k}_{k}}g_{jl} -  \tensor{\delta}{^{k}_{l}}g_{jk} \right)\tensor{q}{_{i}^{j}} \right) \\
          &= -2\gaussianCurvature q_{il} \formPeriod
          \end{aligned}
    \end{align}
    Finally, we get
    \begin{align}
        \int_{\surf}\left\| \Div\qb \right\|^{2}\dS
            &= -\frac{1}{2}\int_{\surf} \left( \tensor{q}{_{il}^{|j}_{|j}} -2 \gaussianCurvature q_{il} \right) q^{il}\dS
             = \int_{\surf} \frac{1}{2}\left\| \nabla\qb \right\|^{2} + \gaussianCurvature\trace\qb^{2}\dS \formComma\\
        \int_{\surf}\left\langle \nabla\qb,  \left( \nabla\qb \right)^{T_{(2\,3)}}\right\rangle\dS
            &= -\int_{\surf} \tensor{q}{_{i}^{k}_{|l|k}} q^{il}\dS
             = -\int_{\surf} \left( \tensor{q}{_{i}^{k}_{|k|l}} + 2\gaussianCurvature q_{il} \right) q^{il}\dS\notag\\
            &= \int_{\surf} \left\| \Div\qb \right\|^{2} -2\gaussianCurvature\trace\qb^{2}\dS
             = \int_{\surf} \frac{1}{2}\left\| \nabla\qb \right\|^{2} - \gaussianCurvature\trace\qb^{2}\dS \formPeriod
    \end{align}

  \end{proof}

  \begin{lem}\label{lem:quadtwotensortheorem}
    For all 2-tensors \( \tb\in\tangent^{(2)}\surf \) at surface \( \surf \) holds
    \begin{align}\label{eq:quadresult}
      \tb^{2} &=  \left( \trace\tb \right)\tb + \frac{1}{2}\left( \trace\tb^{2}-  \left( \trace\tb \right)^{2} \right)\gb \formPeriod
    \end{align}
  \end{lem}
  \begin{proof}
    With the surface Levi-Civita tensor \( \Eb \) defined in \eqref{eq:LeviCivitaTensorSurf},
    the quarter turn in the row and column space of a 2-tensor \( \tb\in\tangent^{(2)}\surf \) is
    \begin{align}\label{eq:quadproof1}
      \left[ \Eb\tb\Eb \right]_{ij}
          &= E_{ik}E_{lj}t^{kl}
           = \left( g_{il}g_{kj} - g_{ij}g_{kl} \right)t^{kl}
           = t_{ji} - \tensor{t}{^{k}_{k}}g_{ij}
           = \left[ \tb^{T} - \left( \trace\tb \right)\gb \right]_{ij} \formPeriod
    \end{align}
    Particularly, \eqref{eq:quadproof1} is also valid for the square of \( \tb \), \ie,
    \begin{align}\label{eq:quadproof2}
      \Eb\tb^{2}\Eb = \left( \tb^{2} \right)^{T} - \left( \trace\tb^{2} \right)\gb\formPeriod
    \end{align}
    On the other hand side, with \( \Eb\Eb=-\gb \), \eqref{eq:quadproof1} and \( (\tb^{T})^{2} = (\tb^{2})^{T} \), we calculate
    \begin{align}
      \begin{aligned}\label{eq:quadproof3}
        \Eb\tb^{2}\Eb &= -\left( \Eb\tb\Eb \right)^{2}
                       = -\left( \tb^{T} - \left( \trace\tb \right)\gb \right)^{2}
                       = -\left(\tb^{T}\right)^{2} + 2\left( \trace\tb \right)\tb^{T} - \left( \trace\tb \right)^{2}\gb\\
                      &= -\left( \tb^{2} \right)^{T} + 2\left( \trace\tb \right)\Eb\tb\Eb + \left( \trace\tb \right)^{2}\gb\formPeriod
      \end{aligned}
    \end{align}
    Averaging identities \eqref{eq:quadproof2} and  \eqref{eq:quadproof3} results in
    \begin{align}\label{eq:quadproof4}
      \Eb\tb^{2}\Eb &= \left( \trace\tb \right)\Eb\tb\Eb + \frac{1}{2}\left(\left( \trace\tb \right)^{2} -  \trace\tb^{2} \right)\gb \formPeriod
    \end{align}
    Finally, we obtain  \eqref{eq:quadresult} by a quarter turn  with \( \Eb \)  in the row and column space of \eqref{eq:quadproof4}.
  \end{proof}

  \begin{lem}\label{lem:dettheorem}
    For all \textbf{full covariant} 2-tensors \( \tb\in\tangent^{0}_{2}\surf \) on surface \( \surf \) holds
    \begin{align}
      \left( \trace\tb \right)^{2} - \trace\tb^{2} = 2 \frac{\det\tb}{\det\gb}= 2\det\tb^{\sharp}\formComma
    \end{align}
    where \( \det \) means the determinant of the matrix proxy.
  \end{lem}
  \begin{proof}
    We can interpret \( \tb \) as its matrix proxy with components \( t_{ij} \) due to the stipulation of the height of the indices.
    Hence, the determinant can be calculated applying the Levi-Civita symbols \( \varepsilon_{ij}\in\{-1,0,1\} \), \ie,
    \begin{align}\label{eq:detproof1}
      \det\tb &= \frac{1}{2}\sum_{i,j,k,l}\varepsilon_{ij}\varepsilon_{kl}t_{ik}t_{jl}\formPeriod
    \end{align}
    With the Levi-Civita tensor defined in \eqref{eq:LeviCivitaTensorSurf}
    we obtain the transformation property
    \begin{align}
      E^{ij} &= \frac{1}{\det\gb}E_{ij} = \frac{1}{\sqrt{\det\gb}}\varepsilon_{ij}\formPeriod
    \end{align}
    Therefore, \eqref{eq:detproof1} results in
    \begin{align}
      \det\tb &= \frac{\det\gb}{2}E^{ij}E^{kl}t_{ik}t_{jl}
               = \frac{\det\gb}{2}\left( g^{ik}g^{jl} - g^{il}g^{jk} \right)t_{ik}t_{jl}
               = \frac{\det\gb}{2}\left( \tensor{t}{_{i}^{i}}\tensor{t}{_{j}^{j}} - \tensor{t}{_{i}^{j}}\tensor{t}{_{j}^{i}} \right)\formPeriod
    \end{align}
    Additionally, we observe
    \begin{align}
      \det\tb^{\sharp} &= \det\left( \tb\cdot\gb^{-1} \right)= \frac{\det\tb}{\det\gb}\formPeriod
    \end{align}
  \end{proof}

  \begin{cor}\label{lem:QuadraticIdenentities}
     For shape operator \( \shapeOperator \) and Q-tensor \( \qb\in\mathcal{Q}(\surf) \) the following identities are valid.
     \begin{align}
       \left\| \shapeOperator \right\|^{2}
            &= \trace\shapeOperator^{2} = \meanCurvature^{2} - 2\gaussianCurvature\label{eq:quadcolresult1}\\
        \shapeOperator^{2}
            &= \meanCurvature\shapeOperator - \gaussianCurvature\gb\label{eq:quadcolresult2}\\
        \left\langle \shapeOperator^{2} , \qb \right\rangle
            &= \meanCurvature\left\langle \shapeOperator, \qb \right\rangle\label{eq:quadcolresult3}\\
        \qb^{2}
            &= \frac{1}{2}(\trace\qb^{2})\gb\label{eq:quadcolresult4}\\
        \left\| \shapeOperator\qb \right\|^{2}
            &= \frac{1}{2}(\trace\qb^{2})\left( \meanCurvature^{2} - 2\gaussianCurvature \right)\label{eq:quadcolresult5}\\
        \trace\left(  \shapeOperator\qb\right)^{2}
            &= \left\langle \shapeOperator, \qb \right\rangle^{2} + \gaussianCurvature\trace\qb^{2}\label{eq:quadcolresult6}
     \end{align}
  \end{cor}
  \begin{proof}
    The proofs here are very straightforward with all the spadework above.
    \eqref{eq:quadcolresult1} is a consequence of \autoref{lem:dettheorem} for \( \shapeOperator\in\tangent^{(2)}\surf \)
    and hence, we obtain also \eqref{eq:quadcolresult2} with \autoref{lem:quadtwotensortheorem}.
    Since  \( \qb\in\mathcal{Q}(\surf) \) is trace-free, we follow from \eqref{eq:quadcolresult2}, that
    \( \left\langle \shapeOperator^{2} , \qb \right\rangle = \meanCurvature\left\langle \shapeOperator, \qb \right\rangle - 2\gaussianCurvature\trace\qb \)
    and therefore \eqref{eq:quadcolresult3}.
    Again, \( \qb \) is a Q-tensor and thus \autoref{lem:quadtwotensortheorem} results in \eqref{eq:quadcolresult4}.
    The shape operator \( \shapeOperator \) is self-adjoint, so with \eqref{eq:quadcolresult4} we can calculate
    \begin{align}
      \left\| \shapeOperator\qb \right\|^{2} &= \left\langle \shapeOperator\qb, \shapeOperator\qb \right\rangle
                    = \left\langle \shapeOperator^{2},\qb^{2} \right\rangle
                    = \frac{1}{2}(\trace\qb^{2})\left\langle \shapeOperator^{2}, \gb \right\rangle
                    = \frac{1}{2}(\trace\qb^{2})\trace\shapeOperator^{2} \formPeriod
    \end{align}
    and get \eqref{eq:quadcolresult5} with \eqref{eq:quadcolresult1}.
    We note that \(\left\langle \shapeOperator, \qb \right\rangle^{2} = \left( \trace\shapeOperator\qb \right)^{2}\).
    Therefore,\autoref{lem:dettheorem} results in \eqref{eq:quadcolresult6}, because
    \begin{align}
      \left( \trace\left( \shapeOperator\qb \right) \right)^{2} - \trace\left(  \shapeOperator\qb\right)^{2}
          &= 2\det\left( \shapeOperator\qb \right)^{\sharp}
           = 2\left( \det\shapeOperator^{\sharp} \right)\left( \det\qb^{\sharp} \right)
           = \gaussianCurvature\left( \left( \trace\qb \right)^{2} - \trace\qb^{2} \right)
           = -\gaussianCurvature\trace\qb^{2} \formPeriod
    \end{align}
  \end{proof}

\begin{lem}
  For the inverse thin film metric \( \Gb^{-1} \) holds
  \begin{align}
    \begin{aligned}\label{eq:InverseMetric}
    G^{ij} &= \left( g^{ik} + \sum_{\frakl = 1}^{\infty} \xi^{\frakl}\left[ \shapeOperator^{\frakl} \right]^{ik} \right)
               \left( \delta_{k}^{j} + \sum_{\frakk = 1}^{\infty} \xi^{\frakk}\tensor{\left[ \shapeOperator^{\frakk} \right]}{_{k}^{j}} \right) \\
             &= \left[\left( \gb + \sum_{\frakk = 1}^{\infty} \xi^{\frakk} \shapeOperator^{\frakk} \right)^{2}\right]^{ij}\formComma\\
    G^{\xi\xi} &= 1 \formComma\\
    G^{i\xi} &= G^{\xi i} = 0 \formPeriod
    \end{aligned}
  \end{align}
\end{lem}
\begin{proof}
  First we define the pure tangential components of the thin film metric tensor as \( \Gb_{t} := \left\{ G_{ij} \right\} \).
  With \( \kronecker = \left\{ \delta^{i}_{j} \right\} \) the Kronecker delta, we can write down in usual matrix notation
  \begin{align}
    \Gb\cdot\Gb^{-1} &=
      \begin{bmatrix}
        \Gb_{t} & O \\
          O      & 1
      \end{bmatrix}
      \cdot
      \begin{bmatrix}
         \left\{ G^{ij} \right\} & \left\{ G^{i\xi} \right\} \\
         \left\{ G^{\xi i} \right\} & G^{\xi\xi}
      \end{bmatrix}
          =
             \begin{bmatrix}
                \kronecker & O \\
                O       & 1
             \end{bmatrix} \formPeriod
  \end{align}
  Thus, we obtain
  \begin{align}
    G^{\xi\xi} &= 1 \formComma\\
    G^{i\xi} &= G^{\xi i} = 0 \formComma\\
    \left\{ G^{ij} \right\} &= \Gb_{t}^{-1}
                         =\left( \gb - \xi\shapeOperator \right)^{-2}
                         = \left( \gb - \xi\shapeOperator \right)^{-1} \cdot \left( \kronecker - \xi\shapeOperator^{\sharp} \right)^{-1}\formPeriod
  \end{align}
  For \( h \) small enough, so that \( \xi\left\| \shapeOperator \right\| \le h\left\| \shapeOperator \right\| < 1 \) and exponent with a dot indicate matrix (endomorphism) power, we can use the Neumann series
  \begin{align}
    \left( \kronecker - \xi\shapeOperator^{\sharp} \right)^{-1} = \kronecker + \sum_{\frakk = 1}^{\infty}\xi^{\frakk}\left( \shapeOperator^{\sharp} \right)^{\cdot\frakk} \formComma
  \end{align}
  and therefore the assertion,
  because
  with \( \shapeOperator^{\frakk}  = \left( \shapeOperator\cdot\gb^{-1} \right)^{\cdot\frakk} \cdot \gb\) we get
  \begin{align}
    \left( \shapeOperator^{\sharp} \right)^{\cdot\frakk} &= \left( \shapeOperator\cdot\gb^{-1} \right)^{\cdot\frakk} = \shapeOperator^{\frakk}\cdot\gb^{-1} = \left( \shapeOperator^{\frakk} \right)^{\sharp}
  \end{align}
  and
  \begin{align}
    \left( \gb - \xi\shapeOperator \right)^{-1}
          &= \left( \left( \kronecker - \xi\shapeOperator^{\sharp} \right)\cdot \gb \right)^{-1}
          = \gb^{-1} \cdot \left( \kronecker - \xi\shapeOperator^{\sharp} \right)^{-1}
  \end{align}
\end{proof}

\begin{lem}
  For the determinant of the thin shell film tensor \( \det \Gb   \) holds
  \begin{align} \label{eq:DetMetricExpand}
     \det \Gb   &= \left( 1 - \xi\meanCurvature + \xi^{2}\gaussianCurvature \right)^{2}\det \gb  \formComma
  \end{align}
\end{lem}
\begin{proof}
  The mixed components are zero, so we get
  \begin{align}
    \det \Gb   &= G_{\xi\xi}\det \Gb_{t}   = \det \Gb_{t}  \formPeriod
  \end{align}
  Now, we define \( \sqrt{\Gb_{t}^{\sharp}}   := \left(\gb - \xi \shapeOperator\right)^{\sharp} \) as a square root of \( \Gb_{t}^{\sharp} \),
  because
  \begin{align}
    \Gb_{t}^{\sharp} &= \left( \left( \gb -\xi\shapeOperator \right)^{2} \right)^{\sharp}
            = \left( \left( \gb -\xi\shapeOperator \right)^{\sharp}\left( \gb -\xi\shapeOperator \right)  \right)^{\sharp}
            = \left( \gb -\xi\shapeOperator \right)^{\sharp}\left( \gb -\xi\shapeOperator \right)^{\sharp}
            = \left( \sqrt{\Gb_{t}^{\sharp}} \right)^{2} \formPeriod
  \end{align}
  Hence, we can calculate
  \begin{align}
    \det \Gb   &= \det \Gb_{t}   = \det \Gb_{t}^{\sharp}\gb
                  =\det \Gb_{t}^{\sharp}   \det \gb
                  =\det \sqrt{\Gb_{t}^{\sharp}}  ^{2} \det \gb   \formPeriod
  \end{align}
  For the determinant of \( \sqrt{\Gb_{t}^{\sharp}} \), we regard that \( \gb^{\sharp}=\kronecker \) is the Kronecker delta,
  so we obtain
  \begin{align}
     \det \sqrt{\Gb_{t}^{\sharp}}
          &= \det\left( \gb^{\sharp} - \xi\shapeOperator^{\sharp}  \right)
           = \left( 1 - \xi\tensor{B}{_{u}^{u}} \right)\left( 1 - \xi\tensor{B}{_{v}^{v}} \right)
                - \xi^{2}\tensor{B}{_{u}^{v}}\tensor{B}{_{v}^{u}} \\
          &= 1 - \xi\left( \tensor{B}{_{u}^{u}} + \tensor{B}{_{v}^{v}} \right)
                + \xi^{2}\left( \tensor{B}{_{u}^{u}} \tensor{B}{_{v}^{v}} - \tensor{B}{_{u}^{v}}\tensor{B}{_{v}^{u}} \right)
           = 1 - \xi\trace\shapeOperator + \xi^{2}\det\shapeOperator^{\sharp}\\
          &=1 - \xi\meanCurvature + \xi^{2}\gaussianCurvature \formPeriod
  \end{align}
\end{proof}

\begin{lem}\label{lem:BCAtSurfaceExpansion}
  Let \( \Wb \) be an arbitrary \( n \)-tensor in the thin film (with sufficient regularity), which vanish at the boundaries, \ie,
  \( \Wb\in\left\{ \Psib\in\tangent^{(n)}\surfh:\, \Psib=0 \text{ at }\partial\surfh \right\} \), holds
  \begin{align}
    \Wb|_{\surf} = \partial_{\xi}\Wb|_{\surf} = \landau(h^{2})\formPeriod
  \end{align}
\end{lem}
\begin{proof}
  We denote the boundary at \( \xi=h/2 \) by \( \Upsilon^{+} \) and \( \Upsilon^{-} \) at \( \xi=-h/2 \), \st\, \( \Upsilon^{+}\cup\Upsilon^{-} = \partial\surfh \).
  Taylor expansions at the surface result in
  \begin{align}
    0 &= \Wb|_{\Upsilon^{\pm}}
       = \Wb|_{\surf} \pm \frac{h}{2}\partial_{\xi} \Wb|_{\surf} + \frac{h^{2}}{8}\partial_{\xi}^{2} \Wb|_{\surf} + \landau(h^{3}) \formPeriod
  \end{align}
  And we yield
  \begin{align}
    0 &= \Wb|_{\Upsilon^{+}} + \Wb|_{\Upsilon^{-}} = 2\Wb|_{\surf} + \landau(h^{2})\\
    0 &= \Wb|_{\Upsilon^{+}} - \Wb|_{\Upsilon^{-}} = h\partial_{\xi}\Wb|_{\surf} + \landau(h^{3})\formPeriod
  \end{align}
\end{proof}

\enlargethispage{20pt}

\bibliographystyle{siam}
\bibliography{bibliography}

\end{document}